\documentclass[a4paper,12pt]{article}

\usepackage{amsthm}
\usepackage{amsmath}
\usepackage{amssymb}
\usepackage{latexsym}
\usepackage{xcolor}
\newtheorem{theorem}{Theorem}[section]
\newtheorem{proposition}[theorem]{Proposition}
\newtheorem{lemma}[theorem]{Lemma}

\newtheorem{remark}[theorem]{Remark}
\newtheorem{definition}[theorem]{Definition}
\newtheorem{example}[theorem]{Example}
\newtheorem{assumption}[theorem]{Assumption}

\begin{document}

\title{On robust fundamental theorems of asset pricing in
	discrete time\thanks{This work was supported by Center for Mathematical Modeling and Data Science, Osaka University, Japan. We thank Masaaki Fukasawa, Mikl\'os R\'asonyi, Martin Schweizer, and the referees for constructive comments on the earlier versions of the paper. 
	}
}


\author{Huy N. Chau\\
	            Department of Mathematics, The University of Manchester\\%         \and
      %  Second Author %etc.
}


%\institute{Huy N. Chau \at
 %             Department of Mathematics, The University of Manchester.\\
  %            \email{huy.chau@manchester.ac.uk}           %  \\
     %      \and
      %     S. Author \at
     %         second address
%}


\maketitle

\begin{abstract}
This paper is devoted to a study of robust fundamental theorems of asset pricing in discrete time and finite horizon settings. Uncertainty is modelled by a (possibly uncountable) family of price processes on the same probability space. Our technical assumption is the continuity of the price processes with respect to uncertain parameters. In this setting, we introduce a new topological framework which allows us to use the classical arguments in arbitrage pricing theory involving $L^p$ spaces, the Hahn-Banach separation theorem and other tools from functional analysis. The first result is the equivalence of a ``no robust arbitrage" condition and the existence of a new ``robust pricing system". The second result shows superhedging dualities and the existence of superhedging strategies without restrictive conditions on payoff functions, unlike other related studies. The third result discusses  completeness in the present robust setting. When other options are available for static trading, we could reduce the set of robust pricing systems and hence the superhedging prices.
%Insert your abstract here. Include keywords, PACS and mathematical subject classification numbers as needed.
%\keywords{Uncertainty \and Arbitrage \and Fundamental theorem of asset pricing \and Super-hedging \and Discrete time}
%\PACS{PACS code1 \and PACS code2 \and more}
%\subclass{MSC 91B28 \and MSC 60G48}
%\noindent\textbf{JEL Classification}  G12  G13 
\end{abstract}

\section{Introduction}
The equivalence between the absence of arbitrage opportunities and the existence of a martingale measure, or the fundamental theorem of asset pricing (FTAP in short), is a core topic to mathematical finance. FTAP results are discussed in classical models under the assumption that the dynamics of risky assets are known precisely, see \cite{kreps1981arbitrage}, \cite{dalang1990equivalent}, \cite{delbaen1994general}, \cite{delbaen2006mathematics}, etc. Nonetheless, model uncertainty, i.e., the risk of using wrong models, cannot be ignored in practice. Since the seminal work of Knight (1921) \cite{knight1921risk}, uncertainty modeling has emerged as effective tools to address this issue. 

The pathwise approach, pioneered by \cite{hobson1998robust}, makes no assumptions on the dynamics of the underlying assets. Instead, the set of all models which are consistent with the prices of observed vanilla options was investigated and bounds on the prices of exotic derivatives were derived. The approach was applied to barrier options in \cite{brown2001robust}, to forward start options in \cite{hobson2012}, to variance options in   \cite{carr2010}, to weighted variance swaps in \cite{davis2014arbitrage},  among others. \cite{davis2007} introduced the concept of model independent
arbitrage and characterized three different situations that a set of option prices would fall into: absence of arbitrage, model-independent arbitrage, or weak forms of model-dependent arbitrage. A notion of weak arbitrage was discussed in  \cite{cox2011robust} to deal with the case of infinitely many given options. In discrete time,  \cite{deparis} proved a duality result for a class of continuous payoffs in a specific topological setup.  Using the theory of Monge–Kantorovich mass transport, \cite{bhp} established superhedging dualities for exotic options. Pathwise versions of FTAP were given in \cite{riedel2015financial} for a one-period market model and in \cite{acciaio2016model} for a continuous time model where a superlinearly growing option is traded. In discrete time markets,  \cite{burzoni2016universal}, \cite{burzoni2019pointwise}  proved versions of FTAP by  investigating different notions of arbitrage and using different sets of admissible scenarios. \cite{burzoni2017} proved a superhedging
duality theorem, characterized
a subset of trajectories for which no duality gap appears. \cite{fahim2016} studied the problem of superhedging under constraints. We   
refer to \cite{hobson2011} for more references on the subject.

The quasi-sure approach assumes that there is a set of priors $\mathcal{P}$ instead of a specific reference probability measure. The first studies are \cite{avellaneda1995pricing}, \cite{Lyons}, where the uncertain volatility model was introduced. In these papers, the unknown volatility process is assumed to evolve in a fixed interval. Optimization problems under multiple priors were studied in \cite{quenez}, \cite{schied06}, \cite{nutz2016} and others. %This results in a new stochastic analysis, see \cite{peng2004nonlinear}, \cite{denis2006theoretical}, \cite{denis2011function}, \cite{soner2011martingale}.  
For a nondominated set of priors $\mathcal{P}$, \cite{bn15} proved a version of FTAP in discrete time, and obtained a family $\mathcal{Q}$ of martingale measures such that each $P \in \mathcal{P}$ is dominated by a martingale measure in $\mathcal{Q}$ which may be nonequivalent to $P$.  Other versions of FTAP and superhedging dualities under constraints were given in \cite{bayraktar2017}. Recently, \cite{blanchard2019no} proved that there is a subset $\mathcal{P}_0 \subset \mathcal{P}$ such that each model $P \in \mathcal{P}_0$ satisfies the classical no arbitrage condition. In a continuous-time setting, \cite{biagini2015robust} introduced a robust notion of no-arbitrage of the first kind and showed that it holds if and only if every $P \in \mathcal{P}$ admits an equivalent local martingale measure 
up to a certain lifetime. In a discrete-time setting, \cite{nutz2014} established the existence of optimal superhedging strategies by using the notions of medial limit, sequential closedness, and by imposing the existence of martingale measures. Under technical conditions, \cite{obloj2018unified} proved that the pathwise approach and the quasi-sure approach are equivalent.

The parametrization approach 
assumes that the uncertain risky assets are modelled by a parametrized family of dynamics $S^{\theta}, \theta \in \Theta$  on the same filtered probability space.  %{\color{red}reminiscent of} the discrepancy between the notions of strong and weak solutions of {\color{red}stochastic differential equations. %More precisely, % The parametrization approach, adopted in the present paper, proposes that stipulates different dynamics for stock prices on the same probability space 
In this framework, \cite{rasonyi2018utility} and \cite{chau2019robust} studied the robust utility maximization problems in discrete time and in  continuous time settings, respectively. However, FTAP results have not been obtained yet and we prove such fundamental results in this paper. We summarize our approach, contributions and more detailed comparisons to related studies below.
\begin{itemize}
\item 
From the modelling point of view, the parametrization framework differs from the pathwise and the quasi-sure approaches in different ways. In the pathwise approach, randomness and filtrations  are generated by the canonical process. The quasi-sure approach works with Polish spaces and filtrations come from universal completion of Borel sigma fields.  In contrast, the parametrization approach does not require any conditions on the state space. It  may incorporate different sources of randomness to each price process and the filtration could be greatly richer than the natural filtration of stock prices. These properties are useful if one wishes to deal with many price processes and more complex payoffs. For example, in a discrete-time  quasi-sure setting where European options are available for static trading, \cite{bhz2015} showed that, in their languages, the superhedging price (hedger’s price) of an American option can be strictly greater than the highest model-based price (Nature’s price). In a similar framework,  \cite{hobson2017} showed that these
two prices are equal. In \cite{bz2017}, the authors explained the two conflicting results by showing the differences in the information used by the hedger and Nature: in \cite{bhz2015} the hedger and Nature have the same information while Nature has access to more information in \cite{hobson2017}. Another explanation was given in \cite{hobson2016}, where the authors argued that \cite{bhz2015} did not  consider enough models.  \cite{bz2017} introduced randomized models which are relevant for computing the
superhedging prices and reconciled the two results.  In another direction, \cite{adot} enlarged the probability space by the exercise times of
the American option to recover the superhedging duality. To sum up, a sufficient amount of information and randomness is required for pricing American options. %for example, \cite{hobson2016}, \cite{hobson2017} showed that one must consider various specifications for the flow of information to find robust bounds for American options. 
%Although we could model the canonical process in very high dimensional settings to incorporate more randomness, we need to be careful when working with infinite dimensional ones. 
In the parametrization framework, we can introduce as many randomness, e.g.,  an infinite number of driving noises, as we wish. This flexibility may be useful for pricing American and other complex options. However, this question is beyond the scope of this paper and thus left for future studies. On the other hand, Polish spaces are not enough  in some situations, see for example  \cite{khan2014}. Theorems 1, 2 of \cite{khan2014} proved that the existence of equilibria in various game settings is ensured if the probability space is saturated, the  concept was first introduced in \cite{hoover1984}. Polish spaces, which are called ordinary probability spaces in \cite{fajardo2017},  are not saturated, see Theorem 3B.1 of \cite{fajardo2017}. Therefore, the present work offers an alternative approach that would be flexible enough for different modelling purposes. The choice between these frameworks is a modelling issue. 

%Secondly, the quasi sure and pathwise approach can be interpreted as ``top down": the canonical filtration is fixed, strategies for each model have to be adapted to the same filtration, namely functions of the canonical process, and consequently, there is no difficulty for defining admissibility. The parametrization is ``bottom up": filtration is determined by driving noises and strategies for each model are adapted to this information flow, and admissibility is a model dependent concept. 
\item The three approaches are also different from technical points of view. The pathwise approach assumes that there are some traded vanilla options from which marginal distributions of the underlying assets are deduced.  Techniques with martingale optimal transports are employed to derive robust  bounds for other exotic options. In the quasi-sure approach, ones have to work in a ``local fashion" where heavy tools from the theory of analytic sets and measurable selections are applied to glue one-period solutions together by using dynamic programming.  In contrast, the parametrization framework does not require the existence of other options as a part of modelling, and it allows to use the standard arguments from the classical no arbitrage pricing theory.  Our proof techniques include a new global argument without dynamic programming. The global argument is suitable for continuous time models and in particular for models with transaction costs, see \cite{cfr2020} for more details. Most importantly, we are able to apply the $L^p$ theory to reach satisfactory results without restrictive conditions. %In particular, we prove a superhedging duality and the existence of superhedging strategies without any restrictive condition on option payoffs, in contrast to most of related literature in uncertainty modelling. 

\item Different versions of the first FTAP were discussed in uncertainty modelling frameworks. In the quasi-sure setting with a nondominated family
of priors $\mathcal{P}$, the Dalang-Morton-Willinger (DMW) theorem was extended in \cite{bn15}. In \cite{bn15},  the probability space $\Omega$ is required to be Polish, or to be a product of Polish spaces, e.g., $\Omega = \Omega_1^T$, where $\Omega_1$ is a Polish space. The set $\mathcal{P}$ contains product measures of the form $P_0 \otimes P_1 \otimes ...  \otimes P_{T-1}$ where $P_t$ is a universally measurable selector of $\mathcal{P}_t$, the set of possible models for the $t$-th period, given
state $\omega$ at time $t$. It is technically imposed that $\text{graph}(\mathcal{P}_t)$ must be an analytic subset of $\Omega_t \times \mathfrak{B}(\Omega_t)$, where $\mathfrak{B}(\Omega_t)$ is the set
of all probability measures on $\Omega_t$. In \cite{acciaio2016model}, a pathwise version of the first FTAP was given, under the existence of a superlinearly growing option. This condition ensures the compactness of the set of martingale measures compatible with option prices. In the parametrization  setting, we prove a robust version of the DMW theorem without any condition on $\Omega$ and the existence of other traded options, thus our results cannot be deduced from \cite{bn15}, \cite{acciaio2016model}. Technically, we assume the continuity of the price processes with respect to the uncertain parameters. In addition, the laws of the uncertain price processes  in the current setting are not necessarily of the product forms, see Example 4.2 of \cite{rasonyi2018utility}.  

\item Different superhedging dualities under uncertainty were introduced. \cite{denis2006} established a duality in a continuous time setting for option payoffs in the class $C_b(\Omega)$, the set of bounded continuous functions on $\Omega$.  The duality of \cite{burzoni2017} holds true for any measurable claim, however, \cite{burzoni2017} imposed the existence of martingale measures and did not study the no arbitrage property. Under transaction costs, \cite{dolinsky14} established a duality result for  upper semicontinuous option payoffs which are at most quadratic growth. Under technical assumptions, \cite{bn15} proved a superhedging duality for upper semianalytic payoffs using a dynamic programming argument. As pointed out by the authors, semianalyticity is preserved through the recursion, whereas Borel-measurability is not. The duality from \cite{acciaio2016model} works for upper semicontinuous and linearly bounded from above payoffs. Other results were obtained in \cite{bhp}, \cite{dolinsky14b}, \cite{dolinsky-soner2},  \cite{hou2018}, and others.  Our superhedging dualities apply to the most general class of option payoffs  as long as a mild  integrability condition is satisfied. Furthermore, we prove the closedness of the set of hedgeable claims and the existence of hedging strategies by using tools from functional analysis. Our arguments are completely different from that of \cite{bn15}, \cite{acciaio2016model}, and suitable for continuous time settings, see also \cite{cfr2020}.   
\item There are many situations where the parametrization approach is not applicable.
For example, when the investors do not know precisely what null and sure events are, the quasi-sure approach is appropriate. The parametrization framework is not suitable in such situations because these events are known to the investors under a fixed probability measure. Other examples could be found in \cite{rasonyi2018utility}.

\end{itemize}

% This feature will be seen clearly in situations with continuous time models where 

%Instead, we start with a , which offers an alternative way to investigate continuous time models or to situations where local properties do not imply global ones, such as the no arbitrage property in markets with transaction costs, see \cite{cfr2020}.

The paper is organized as follows. Section \ref{sec:model} introduces the setting. Main results and discussion are given in Section \ref{sec:main_results}. Proofs are given in Section \ref{sec:main}. Some preliminaries and useful results are given in Section \ref{sec:app}.

\textit{Notations}. Let $I$ be an arbitrary set. In the product space $\mathbf{X}=\prod_{i \in I} X_i$, a vector $(f^i)_{i \in I}$ will be denoted by $\mathbf{f}$. We write $\mathbf{f} \ge \mathbf{g}$ if $f^i \ge g^i$ for all $i \in I$. In addition, $\mathbf{1}$ denotes the vector with all coordinates equal to $1$ and  $\mathbf{1}^{i}$ denotes the vector with only the coordinate $i$ equals to $1$ and the others are zero. 

\section{The model}\label{sec:model}
In this paper, we work with a discrete time framework. Let $(\Omega, \mathcal{F}, (\mathcal{F}_{t})_{t = 0,1,...,T}, P)$ be a filtered probability space and $T \in \mathbb{N}$. We suppose that $\mathcal{F}_0$ contains
all $P$-zero sets. Let  $B \equiv 1$ be a non-risky asset. Let $\Theta$ be a non-empty set, interpreted as the parametrisation of uncertainty. For each $\theta \in \Theta$, let $S^{\theta}_t = (S^{\theta,1}_t,...,S^{\theta,d}_t), d \ge 1, 0 \le t \le T$ be a $\mathbb{R}^d$-valued process such that $S^{\theta}_t$ is $\mathcal{F}_t$-measurable and  $S^{\theta}_0 \in \mathbb{R}^d$. The increments of a process $Y$ are denoted by $\Delta Y_t := Y_t - Y_{t-1}, t = 1,..., T.$ In order to include interesting robust models,  e.g., models with a nondominated family of laws, $\Theta$ is neither countable nor convex. Suppose that the true risky asset is $S= S^{\theta^*}$, where $\theta^* \in \Theta$ is the exact, but unknown, parameter. 
\begin{assumption}\label{assum:Theta}
	The following conditions are imposed throughout the paper:
	\begin{itemize}
		\item[(i)] $\Theta$ is a subset of a separable metric space,
		\item[(ii)] For any sequence $(\theta_n)_{n \in \mathbb{N}}$ such that $\theta_n \to \theta$ in $\Theta$, there exists a subsequence $(\theta_{n_k})_{k \in \mathbb{N}}$ such that for each $0\le t \le T$, it holds that $\lim_{k \to \infty} S^{\theta_{n_k}}_t = S^{\theta}_t, a.s. $ 
	\end{itemize} 
\end{assumption}
Assumption \ref{assum:Theta} (i) requires a metric structure for $\Theta$ and $\Theta$ is also separable, as a subset of a separable metric space, see Theorem 15.13 of \cite{Schechter1996handbook}. Assumption \ref{assum:Theta} (ii) asks for the continuity of the risky asset with respect to the uncertain parameter. For example, if $S^{\theta_n}_t$ converges to $S^{\theta}_t$ in $L^p, p \ge 1,$ when  $\theta_n$ converges to $\theta$, Assumption \ref{assum:Theta} (ii) holds. When $\Theta$ is uncountable, these technical conditions are necessary for our arguments to extend the proofs for countable cases to uncountable ones.  Furthermore, the continuity  in Assumption \ref{assum:Theta} does not necessarily imply that the laws of $S^{\theta}, \theta \in \Theta$ are dominated, see Example \ref{ex:nondominated}.  As discussed in Section 3.2.1 of \cite{lms}, the present modelling framework corresponds to the class $\mathbf{S}$ property introduced in their paper.
\begin{example}[A toy model]\label{ex:toy}
Consider the $t$-fold Cartesian product $\Omega_t = \{ -1, 1\}^t$. Denote by $2^{\Omega_t}$ the power set of $\Omega_t$ and $\Omega = \Omega_T$. The mapping $\Pi_t: \Omega_T \to \Omega_t$ is defined by
	\begin{equation}\label{eq:Pi}
		\Pi_t(\omega):= (\omega_1,...,\omega_t), \qquad \forall \omega = (\omega_1,...,\omega_T) \in \Omega.
	\end{equation}
	%$$\Pi_t(\omega):= (\omega_1,...,\omega_t), \qquad \forall \omega = (\omega_1,...,\omega_T) \in \Omega.$$
	Set $\mathcal{F}_t = \Pi^{-1}_t(2^{\Omega_t})$. The measurable space $(\Omega, 2^{\Omega})$ is equipped with the probability measure
	$$P(\{\omega\}) = \prod_{t=1}^T \left( \frac{1}{2} \delta_{1}(\{ \text{proj}_t(\omega) \}) + \frac{1}{2} \delta_{-1}(\{ \text{proj}_t(\omega) \}) \right),$$
	where $\delta_{x}: \mathcal{B}(\mathbb{R}^d) \to \{0,1\}$ is the Dirac measure at $x$
	\begin{eqnarray*}
		\delta_x(B) = \left\{
		\begin{array}{ll}
			1  & \mbox{if } x \in B, \\
			0 & \mbox{otherwise}
		\end{array}
		\right.
	\end{eqnarray*}
	and $\text{proj}_t$ is the projection at $t$ from $\Omega$ to $\{-1,1\}$,
	$$ \text{proj}_t(\omega_1,...,\omega_T) = \omega_t, \text{ for } \omega \in \Omega.$$
	For a vector of parameters $\theta = \left( (\mu_{t})_{1 \le t \le T}, (\sigma_{t})_{1 \le t \le T} \right) \in \Theta = [\underline{\mu},\overline{\mu}]^T \times [\underline{\sigma},\overline{\sigma}]^T$ for some $0 < \underline{\sigma} <\overline{\sigma}, \underline{\mu},\overline{\mu} \in \mathbb{R}$, we define the real-valued process $S^{\theta}$ as
	\begin{equation}\label{eq:S}
		S^{\theta}_t = s_0 + \sum_{u=1}^t (\mu_{u} + \sigma_{u}\text{proj}_u(\omega) ), \qquad s_0 \in \mathbb{R}^d.
	\end{equation}
	Here the drift and volatility parameters may vary in different periods. This is an example where the volatility and drift are uncertain. 
	\end{example}
\begin{example}[A nondominated family of laws]\label{ex:nondominated}
Let $\Omega_1 = \{-1,1\}$ be equipped with the discrete $\sigma$-algebra and the uniform measure $P_1$. We construct the product space $\Omega = \Omega_1^{\kappa}$ together with the product probability measure $P$, where $\kappa$ is an uncountable cardinal, see Section 2.4 of \cite{tao}. As discussed in \cite{fajardo2017}, \cite{Keisler2009}, the probability space $\Omega$ is saturated, however, not a Polish space. Assume $T=1$. Let $Y : \Omega \to \mathbb{R} $ be a random variable such that $P(Y > 0)P(Y < 0) > 0$.  Define 
$$S^{\theta}_1 = 1 + \theta Y(\omega), \qquad  S^{\theta}_0 = 1,$$
where $\Theta = [\underline{\sigma}, \overline{\sigma}] \subset \mathbb{R}_+$. Assumption \ref{assum:Theta} is fulfilled. Denote by $P^{\theta}(\cdot):=P(S^{\theta} \in \cdot)$ the law of $S^{\theta}$, and $A^{\theta} = \{ 1+ \theta Y(1), 1 + \theta Y(-1)\}$ so that $P^{\theta}(A^{\theta}) = 1$, for $\theta \in \Theta$. The family $\mathcal{P}= \{P^{\theta}, \theta \in \Theta\}$ is not dominated by a probability measure. Indeed, if it is dominated, the Halmos Savage lemma, see Lemma 7 of \cite{halmossavage}, implies the existence of a countable subclass $\{P^{\theta_i}, i \in \mathbb{N}\} \approx \mathcal{P}$, in the sense that $\sup_{i\in \mathbb{N}} P^{\theta_i}(A)= 0 \Leftrightarrow \sup_{P \in \mathcal{P}}P(A) = 0$. This contradicts the fact that $P^{\theta}(A^{\theta}) = 1 \ne P^{\theta'}(A^{\theta})=0$ when $\theta \ne \theta'$.  
\end{example}
\begin{example}
Let $\Theta = \prod_{i=1}^d \prod_{j=1}^n \left(  [\underline{\mu}^i,\overline{\mu}^i] \times [\underline{\sigma}^{ij},\overline{\sigma}^{ij}]\right) $	for some $0 < \underline{\sigma}^{ij} <\overline{\sigma}^{ij}, \underline{\mu}^i,\overline{\mu}^i \in \mathbb{R}$. For each $\theta = (\mu^i,\sigma^{ij})_{1\le i \le d, 1 \le j \le n} \in \Theta$, we consider 
$$S^i_{t} = S^i_{t-1}\left( \mu^i \Delta t + \sum_{j=1}^n\sigma^{ij} \Delta W^j_t\right),$$ 
where $\Delta t > 0$ and $\Delta W^j_t, j \in \{1,...,n\}, t \in \{1,...,T\}$ are independent normally distributed random variables with mean $0$ and variance $\Delta t$. In this example, each risky process follows a discrete
version of the Black–Scholes model where the drift $\mu^i$ and the volatilities $\sigma^{i,j}$ are uncertain. 
\end{example}
\begin{example}[Example 4.4 of \cite{rasonyi2018utility}] 
We recall a nonparametric example of \cite{rasonyi2018utility}. On a probability space $(\Omega, \mathcal{F}, P)$, we consider independent
and standard uniform random variables $\varepsilon_1,...,\varepsilon_T$. Denote the filtration $\mathcal{F}_t = \sigma(\varepsilon_1,...,\varepsilon_t), t = 1,...,T$. The space $C([0,1])$ of all continuous real-valued functions on $[0,1]$ with the metric of uniform convergence is a separable space. Let $\Theta = \{\theta = (\theta_1,..,\theta_T)\}$ be a subset of $C([0,1])^T$ such that for each $t =1,...,T$, the function $\theta_t: [0,1] \to \mathbb{R}$ is nondecreasing. Define
$$S^{\theta}_t(\omega) = s_0 + \sum_{s = 1}^t \theta_s(\varepsilon_s(\omega)),$$
for some $s_0 \in \mathbb{R}$. Assumption \ref{assum:Theta} is satisfied. This example covers many kinds of distributions for $S^{\theta}$ because any distribution can be generated from its cumulative distribution function and a standard uniform random variable by the inversion method. 
\end{example}
\begin{example}We consider the weak identification framework discussed in \cite{li}. Let $L$ be the lag operator. Assume that $y_t = \log(S_t/S_{t-1})$ is modelled as an $ARFIMA(1, d, 0)$ process
$$(1-\alpha L)y_t =  \sigma(1-L)^d \varepsilon_t,$$
where $\alpha
$ is the autoregressive coefficient, $d$ is the memory parameter, $\varepsilon_t$ is a stationary martingale difference sequence with unit variance. It is explained in \cite{li} that the $ARFIMA(1, d, 0)$ model with $\alpha = 1$ is
equivalent to the $ARFIMA(1, d + 1, 0)$ model with $\alpha = 0$. In particular, a ``rough" configuration with $(d \approx - 0.5, \alpha \approx 1)$ is observationally similar to a ``long memory" configuration with $(d \approx 0.5, \alpha \approx 0)$. Consider $\Theta = \left([\underline{d},0] \times  [1-\delta, 1] \right) \cup \left( [0, \overline{d}] \times [-\delta, \delta]  \right) $ where $\delta > 0, \overline{d} > 0, \underline{d} < 0$ are real constants.  The case $\theta = (d,\alpha) \in [\underline{d},0] \times  [1-\delta, 1] $ corresponds to the rough configuration and the case $\theta = (d, \alpha) \in [0, \overline{d}] \times [-\delta, \delta]$ corresponds to the long memory one. This is a robust version of the framework of \cite{li}. The set $\Theta$ is \emph{not} convex, however, is fitted in our robust setting. 
\end{example}

Our framework is applicable in many situations where other approaches are not. In Example \ref{ex:toy}, denote by $P^{\theta}(\cdot):=P(S^{\theta} \in \cdot)$ the law of $S^{\theta}$ and $\mathcal{P} = \{P^{\theta}, \theta \in \Theta\}$. The set $\mathcal{P}$ is not convex in general, hence does not satisfy the conditions of the set of probability measures $\mathcal{P}_t(\omega)$ in \cite{bn15}, \cite{barlt}. In Example \ref{ex:nondominated}, $\Omega$ is not a Polish space.  In Example 4.2 of \cite{rasonyi2018utility}, the “time-consistency” property in \cite{bn15}, \cite{nutz2016}, etc. is not satisfied. However, these examples could be treated in our framework.  On the other hand, there are situations where our framework is not useful, see Example 4.3 of \cite{rasonyi2018utility}. 

%\begin{remark}
%In the quasi-sure framework, ones usually choose a family of possible dynamics $S^{\theta}, \theta \in \Theta$ for the risky assets and consider the set of their laws $\mathcal{P} = \{ \text{Law}(S^{\theta}), \theta \in \Theta \}$ on path spaces, see for example \cite{avellaneda1995pricing}, \cite{denis2006}. Assumption \ref{assum:Theta} (ii) is close in spirit to the conditions imposed in these papers.
%\end{remark}

In the following, we introduce our new techniques. In classical settings without uncertainty, one strategy generates one attainable payoff. In the present robust setting, for a given strategy $H$ there is a corresponding family of attainable payoffs $\{H \cdot S^{\theta}_T, \theta \in \Theta \}$. Hence the classical arguments for $L^p$ spaces could not work with such family of payoffs under uncertainty.  We show below that the correct function space is the product space $\mathbf{L}^p = \prod_{\theta \in \Theta} L^p$, which is huge, since $\Theta$ is usually an uncountable set. This is the starting point of the present paper.

%it is proposed to consider the family of attainable payoffs as a subset of some $L^p$ spaces where separation theorems are applied. However, this argument is designed for the case when one strategy generates one payoff, and it is not able to capture such family of payoffs under uncertainty.

The space $L^0(\mathcal{F}_t,P,\mathbb{R}^d)$ is equipped with the topology of convergence in probability, induced by the translation-invariant metric $d(f,g):= E[1 \wedge |f - g|]$. With this structure, $L^0(\mathcal{F}_t,P,\mathbb{R}^d)$ is a Fr\'echet space. When there are no confusion we write $L^0(\mathcal{F}_t,P)$ instead of $L^0(\mathcal{F}_t,P,\mathbb{R}^d)$. Define the product space $\mathbf{L}^0(\mathcal{F}_t,P):=\prod_{\Theta} \left( L^0(\mathcal{F}_t,P),d\right) $ with the corresponding product topology. Similarly, for $p \ge 1$, the space $L^p(\mathcal{F}_t,P)$ is equipped with the topology from the norm $\|\cdot\|_p$ and we define $\mathbf{L}^p(\mathcal{F}_t,P) :=\prod_{\Theta} \left( L^p(\mathcal{F}_t,P), \|\cdot\|_p \right)$ with the corresponding product topology. See also Appendix \ref{sec:app_prod} for basic properties of product spaces and their duals. For $t = 0,...,T$, we define by $\mathbf{S}_t = (S^{\theta}_t)_{\theta \in \Theta}$ a vector in $\mathbf{L}^0(\mathcal{F}_t,P)$. Let $\mathcal{A}$ be the set of all predictable processes $H = (H_t)_{t\in\{1,...,T\}}$ where $H_t \in L^0(\mathcal{F}_{t-1},P), t = 1,...,T$, i.e., trading strategies.  For $H \in \mathcal{A}$, we denote 
\begin{eqnarray*}
H \cdot S^{\theta}_t &=& \sum_{s = 1}^t H_s \Delta S^{\theta}_s, \ \theta \in \Theta, t = 1,...,T,\\
 H \cdot \mathbf{S}_t &=& \sum_{s = 1}^t H_s \Delta \mathbf{S}_s,\ t = 1,...,T,
\end{eqnarray*}
where $\Delta \mathbf{S}_s = \mathbf{S}_s - \mathbf{S}_{s-1} = (S^{\theta}_s - S^{\theta}_{s-1})_{\theta \in \Theta}$.
\begin{definition}\label{defi:NRA}
A robust arbitrage opportunity is a self-financing strategy $H \in \mathcal{A}$ such that
\begin{equation*}\label{eq:RA}
	\forall \theta \in \Theta,  H \cdot S^{\theta}_T \ge 0,\ P-a.s. \qquad \text{ and } \qquad  P(H \cdot S^{\theta}_T) > 0, \ \text{ for some } \theta \in \Theta. 
\end{equation*}	
We say that the market satisfies the condition No Robust Arbitrage $(NRA(\Theta))$ if for every self-financing strategy $H \in \mathcal{A}$, 
	\begin{equation}\label{eq:NRA}
	\forall \theta \in \Theta,  H \cdot S^{\theta}_T \ge 0 \text{  }P-a.s. \qquad \text{ then } \qquad \forall \theta \in \Theta,  H \cdot S^{\theta}_T = 0, \ P-a.s.
	\end{equation}
\end{definition}
The property (\ref{eq:NRA}) is rewritten shortly as 
$$ H \cdot \mathbf{S}_T \ge 0 \qquad \text{ then } \qquad  H \cdot \mathbf{S}_T = 0,$$
and it should be noticed that the inequality and equality are $\theta$-wise. 

Denote by $\mathbf{L}^0_+(\mathcal{F}_T,P)$ the non-negative orthant of $\mathbf{L}^0(\mathcal{F}_T,P)$, i.e., the subset of $\mathbf{L}^0$ consisting elements $\mathbf{f} = (f^{\theta})_{\theta \in \Theta}$ such that $f^{\theta}$ is $\mathcal{F}_T$-measurable and $P[f^{\theta}<0]=0$ for all $\theta \in  \Theta$. Let
\begin{eqnarray*}
	\mathbf{K}^{\Theta} &=& \left\lbrace  H \cdot \mathbf{S}_T \in \mathbf{L}^0(\mathcal{F}_T,P), H \in \mathcal{A} \right\rbrace ,\\
	\mathbf{C}^{\Theta} &=& \mathbf{K}^{\Theta} - \mathbf{L}^0_+(\mathcal{F}_T,P),
\end{eqnarray*}
be the sets of all attainable and hedgeable payoffs \footnote{$A - B := \{ a-b, a \in A, b \in B\}$ for two sets $A$, $B$.}, respectively. With these notations, $NRA(\Theta)$ can be formulated as follows
$$ \mathbf{K}^{\Theta} \cap \mathbf{L}^0_+(\mathcal{F}_T,P) = \{\mathbf{0}\} \text{  or equivalently  } \mathbf{C}^{\Theta} \cap \mathbf{L}^0_+(\mathcal{F}_T,P)  = \{\mathbf{0}\}.$$
It is easy to see that $NRA(\Theta)$ reduces to the classical no arbitrage property when $\Theta = \{\theta\}$ is a singleton. We write $NA(\{\theta\})$ for the classical no arbitrage  condition for the model $\theta$.

Some no-arbitrage conditions introduced in \cite{blanchard2019no} are recalled by using the current notations.
\begin{definition}\label{defi:swNA}
	The condition strong no-arbitrage $sNA(\Theta)$ holds true if the condition $NA(\{\theta\})$ holds true for all $\theta \in \Theta$. The condition weak no-arbitrage $wNA(\Theta)$ holds true if there exists some $\theta \in \Theta$ such that the condition $NA(\{\theta\})$ holds true.
\end{definition}
In the present framework, it is possible that $NRA(\Theta)$ holds, however, both $sNA(\Theta)$ and $wNA(\Theta)$ fail, see Case 1 of Example \ref{ex:1}. In the quasi-sure setting with a convex family $\mathcal
P$, the condition $NA(\mathcal{P})$ implies that $wNA(\mathcal{P})$ holds true, see Theorem 3.6. of \cite{blanchard2019no}, although $sNA(\mathcal{P})$ may fail, see also Example \ref{ex:1} for more details.

%Similar to the classical arguments, we will prove the closedness of $\mathbf{C}$ in suitable topological spaces and then apply the Hahn-Banach theorem. 
The typical approach to derive martingale measures is to use the separation theorem. Here, our situation is different since the product spaces $\mathbf{L}^p, p \ge 1$ are not Banach spaces, unless $|\Theta|$ is finite, and even not metric spaces, unless $|\Theta|$ is countable. However, the product spaces $\mathbf{L}^p, p \ge 1$ are locally convex and therefore the Hahn-Banach theorem is applicable. If one wishes to use separation arguments on the product spaces, the closedness of the set of hedgeable claims $\mathbf{C}^{\Theta}$ needs to be obtained first. Crucial differences between sequential closedness and topological closedness in $\mathbf{L}^p, p \ge 1$ lead to the use of nets, instead of sequences, to determine such closures. After establishing the closedness of $\mathbf{C}^{\Theta}$, the Hahn-Banach theorem will give new pricing systems which play the same roles as martingale measures in classical settings. The novelty is that the new pricing systems average jointly scenarios and models to rule out model independent arbitrage opportunities.  

\begin{remark}{\label{re:one}}
	This robust framework is different from the usual setting with multiple assets. In a financial market with $|\Theta|$ underlying assets, a strategy at time $t$ is a vector $(H^{\theta}_{t})_{\theta \in \Theta} \in \prod_{\theta \in \Theta}L^0(\mathcal{F}_{t-1},P)$. In our robust setting, \emph{one} strategy $H_{t} \in L^0(\mathcal{F}_{t-1},P)$  is used for \emph{all} price processes. The set of strategies for the robust setting is much smaller, and as a result, the dual space is much larger. The discrepancy becomes significant when $\Theta$ is uncountable. 
\end{remark}
\begin{remark}
	%In classical settings without uncertainty, there is a filtered probability space $(\Omega, (\mathcal{F}_t)_{0\le t \le T}, P)$ carrying an adapted process $S$ which is interpreted as the (future) dynamics of the stock. 
The dynamics $S^{\theta}, \theta \ne \theta^*$ are not observed from market data, only the stock price $S = S^{\theta^*}$ is observable, but the information is not enough to determine $\theta^*$. All filtrations generated by $S^{\theta}, \theta \in \Theta$ are contained in $\mathcal{F}$. Investors can build strategies as functions of $S^{\theta^*}$, or of $S^{\theta}, \theta \ne \theta^*$, as long as their strategies are $\mathcal{F}$-predictable. %However, this does not imply that investors are able to use model-dependent strategies. 
Let us consider the following example. Assume that $\mathcal{F}$ is generated by a process $Y$ and $S^{\theta}_t =F(Y_t, \theta)$ for unknown parameter $\theta \in \Theta$, where $F$ is a non-linear function. A strategy $H = (H_t)_{t \in \{1,...,T\}}$ could be path-dependent, i.e., $H_t = H_t(Y_{0},...,Y_{t-1})$, or in particular, a function of stock prices $H_t = H_t(S^{\theta^*}_{0},...,S^{\theta^*}_{t-1}) = H_t(F(Y_{0}, \theta^*), ...,F(Y_{t-1}, \theta^*))$, but this does not necessarily mean that $\theta^*$ is revealed.  Profits and losses are modelled by the vector $(H_t (S^{\theta}_{t+1} - S^{\theta}_{t}))_{\theta \in \Theta}$. %Model-dependent strategies, e.g., vectors of the form $(H^{\theta}_t)_{\theta \in \Theta, t \in \{1,...,T\}}$, cannot be used obviously. 
See also examples from \cite{cfr2020}. 
\end{remark} 
\begin{remark}
	Our setting is also different from ``large financial'' market models, where there is a continuum of securities. For example, in bond markets, zero-coupon bonds are parametrized by their maturities $\theta$ which is a continuous parameter. However, only a finite number of bonds are traded at the same time, see \cite{klein2000fundamental}, \cite{rasonyi2003equivalent}, \cite{balbas2007infinitely}. 
\end{remark}
\begin{remark}
	By postulating an appropriate weak topology on each $L^p$, we define the corresponding weak topology on the product space $\mathbf{L}^p= \prod_{\theta \in \Theta} L^p$,  which is useful in continuous time settings, where ``local fashions", see for example \cite{bn15}, are not applicable. We refer to \cite{cfr2020} for a treatment with transaction costs in continuous time settings. 
\end{remark}
%\begin{remark}Our conditions in Assumption \ref{assum:Theta} seem comparable to the ones in the quasi-sure setting. In the quasi-sure approach, it is typically assumed that the probability space $\Omega$ has to be Polish and the set of priors, which is a subset of the Polish space of all probability measures on $\Omega$, is analytic. In our setting, we do not impose conditions on the probability space $\Omega$, and require that the set of uncertain parameters $\Theta$ is a closed subset of a separable metric space. Furthermore, the continuity condition (ii) in Assumption \ref{assum:Theta} is natural.\end{remark}

%\begin{remark}\label{re:product_measure}
%	It is not convenient to consider a product measure on the product space, since a common $\omega$ is used for $(S^{\theta}_T(\omega))_{\theta \in \Theta}$ instead of  $(S^{\theta}_T(\omega^{\theta}))_{\theta \in \Theta}$, which is a different object. ????????????????The current work does not incorporate learning, and this is not matter in some situations, see Subsection \ref{sec:dynamic}. However, it would be interesting to have learning techniques in future work. 
%\end{remark}
It is possible to translate the market structure $(\Omega, (\mathcal{F}_t)_{t = 0,...,T}, P, (S^{\theta})_{\theta \in \Theta})$  into the quasi-sure setting. Ones who wish to do so may follow the general framework of \cite{bn15} for $\widehat{\Omega}_1 = \mathbb{R}^d$ together with the Borel sigma field $\mathcal{B}(\mathbb{R}^d)$. Define $\widehat{\Omega}_t := (\mathbb{R}^d)^t$ with the convention that $\widehat{\Omega}_0$ is a singleton and $\widehat{\Omega} = \widehat{\Omega}_T$. Let $\widehat{S}$ be the canonical process $\widehat{S}_{t}(\widehat{\omega}) = \widehat{\omega}_t$ for $0\le t \le T$ and $(\widehat{\mathcal{F}}_t)_{t = 0,...,T}$ be the canonical filtration. For each $\theta \in \Theta$, define the probability measure $P^{\theta} = Law((S^{\theta}_t)_{0\le t\le T})$.  The convex hull of the set $\{P^{\theta}, \theta \in \Theta\}$ is defined by $\mathcal{P}$, which is a family of priors in the quasi-sure setting. Here the risky asset is identical with the canonical process while \cite{bn15} considered the general case where the risky asset is a measurable function of the canonical process. %There are modelling issues that distinguish the two approaches. First, the measures $\widehat{P}^{\theta}, \theta \in \Theta,$ may be not of product forms.  Secondly, the filtration  $(\widehat{\mathcal{F}}_t)_{t = 0,...,T}$ contains only the information of $\widehat{S}$ while the filtration $(\mathcal{F}_t)_{t = 0,...,T}$ is typically richer than that, for example, $\mathcal{F}_t$ could be generated by an uncountable number of driving noises. Consequently, trading strategies in the quasi-sure setting are limited. Other methods can be used to translate the present setting into the quasi-sure one, nevertheless, similar problems arise. %when the filtration is generated by an uncountable number of driving noises $\omega^{i}, i \in I$, the quasi-sure approach needs to work with $R^{I}$, which is not a Polish space.

%Other differences emerge if ones try to interpret the dependence of trading strategies and wealth processes with respect to uncertainty. In the quasi-sure setting, let us consider an event $\widehat{\omega}$ and a strategy $\widehat{H}(\widehat{\omega})$. If $\widehat{\omega}$ happens under $\widehat{P}^{\theta_i}$ but not under $\widehat{P}^{\theta_j}$ with $j \ne i$, then the strategy $\widehat{H}(\widehat{\omega})$ is executed under $\widehat{P}^{\theta_i}$ and not under $\widehat{P}^{\theta_j}$. In this aspect, strategies are model-dependent. On the other hand, if $\widehat{\omega}$ happens under both $\widehat{P}^{\theta_i}$ and $\widehat{P}^{\theta_j}$, the wealth process $\widehat{H}(\widehat{\omega}) \cdot \widehat{S}(\widehat{\omega})$ is the same for $\widehat{P}^{\theta_i}$ and $\widehat{P}^{\theta_j}$, i.e., model-independent. %The pathwise approach is a special case of the quasi-sure approach when all probability measures on $\widehat{\Omega}$ are considered. The situations are reversed in the parametrization approach. When an event $\omega$ happens, investors do not know which $S^{\theta}$ realizes the price and they cannot choose the model to trade with. Hence, strategies are model-independent. Nevertheless, wealth processes $H(\omega) \cdot S^{\theta}(\omega)$ are model-dependent. 

%The two modelling approaches also differ in terms of arbitrage. 
\section{Main results}\label{sec:main_results}
Recall that $\mathbf{L}^{1}(\mathcal{F}_{T},P)= \prod_{\theta \in \Theta} L^{1}(\mathcal{F}_{T},P)$ is the product space where each $L^{1}(\mathcal{F}_{T},P)$ is equipped with the usual $\|\cdot\|_1$-norm  topology. Define also the direct sum $\bigoplus_{\theta \in \Theta} L^{\infty}(\mathcal{F}_T,P)$. We refer to Appendix for more details on direct sums, product spaces and their duals. The duality $(\mathbf{L}^{1}(\mathcal{F}_{T},P))^* = \bigoplus_{\theta \in \Theta} L^{\infty}(\mathcal{F}_T,P)$ allows us to identify a linear continuous function $\mathbf{Q}: \mathbf{L}^{1}(\mathcal{F}_{T},P) \to \mathbb{R}$ with a vector $(Z^{\theta}_T)_{\theta \in \Theta}$ with only finite number of non-zero elements.
\subsection{The case without options}
 %Such linear continuous function is called strictly positive if $Z^{\theta_i}_T \ge 0, a.s., i=1,...,k$ and $Z^{\theta_i}_T > 0, a.s.$ for some $i \in \{1,...,k\}$. 
The notion of robust pricing system below is central of this paper.
\begin{definition}\label{def:pricing_func}
	A robust pricing system for $\mathbf{S}$ is a linear continuous function $\mathbf{Q}: \mathbf{L}^{1}(\mathcal{F}_{T},P) \to \mathbb{R}$, that is $\mathbf{Q} = (Z^{\theta}_T)_{\theta \in \Theta} \in \bigoplus_{\theta \in \Theta} L^{\infty}(\mathcal{F}_T,P)$ and such that
	\begin{itemize}
		\item[(i)] $0 \le Z^{\theta}_T$ for all $\theta \in \Theta$ and $\mathbf{Q}(\mathbf{1}) = 1$. 
		\item [(ii)]  $\mathbf{S}$ is a generalized martingale under $\mathbf{Q}$, see Definition \ref{defi:mart}, that is 
		$$\mathbf{Q}(1_{A_{t-1}}\mathbf{S}_{t}) = \mathbf{Q}(1_{A_{t-1}}\mathbf{S}_{t-1})$$
		for all $A_{t-1}\in \mathcal{F}_{t-1}, 1\le t\le T.$
	\end{itemize}
\end{definition}
Since $\mathbf{Q}(\mathbf{1}) = 1$ and $Z^{\theta}_T \ge 0$ for all $\theta \in \Theta$, it holds that $E[Z^{\theta}_T1_A] > 0$ for some $ \theta \in \Theta$ and for some $A \in \mathcal{F}_T$. 
For each $\theta \in \Theta$, we denote by 
\begin{eqnarray*}
	\mathcal{Q}^{\theta}:=  \left\lbrace  \mathbf{Q}: \mathbf{Q} \text{ is a robust pricing system for $\mathbf{S}$  and } Z^{ \theta}_T \ne 0 \right\rbrace
\end{eqnarray*}
the set of robust pricing systems for the model $ \theta$ and by $\mathcal{Q} = \bigcup_{\theta \in \Theta} \mathcal{Q}^{ \theta}$ the set of all robust pricing systems. 

Our first result is a robust version of the celebrated Dalang–Morton–Willinger (DMW) theorem, see \cite{dalang1990equivalent}.   
\begin{theorem}[Robust FTAP]\label{thm:robustftap}
	The following are equivalent
	\begin{itemize}
		\item[(i)] $NRA(\Theta)$ holds;
		\item[(ii)] For every $\theta \in \Theta$ and for every $A \in \mathcal{F}_T$ such that $P[A] > 0$, there exists $\mathbf{Q} \in \mathcal{Q}^{\theta}$ such that $\mathbf{Q}(1_A\mathbf{1}^{\theta}) > 0$.
	\end{itemize}
\end{theorem}
Let $\mathbf{f} = (f^{\theta})_{\theta \in \Theta}$ be a random variable in $\mathbf{L}^0(\mathcal{F}_T,P)$. Denote the superhedging price of $\mathbf{f}$ by
$$\pi(\mathbf{f}):= \inf\{ x \in \mathbb{R}: \exists H \in \mathcal{A} \text{ such that } x + H \cdot  \mathbf{S}_T \ge \mathbf{f}, a.s. \}.$$
Assumption \ref{assum:Theta} (ii) does not imply that the super-hedging price $\pi(\mathbf{f})$ can be fully characterised by computing robust superhedging prices of  $(f^{\theta})_{\theta \in \Gamma}$ with respect to any countable set $\Gamma \subset \Theta$, because in general, there is no continuity of the payoff $f^{\theta}$ with respect to $\theta$. The computation of $\pi(\mathbf{f})$ is therefore not trivial.  The following superhedging duality is  our second main result.

\begin{theorem}\label{thm:superhedge}
	Let $NRA(\Theta)$ hold. Let $\mathbf{f} \in \mathbf{L}^0(\mathcal{F}_T,P)$ be a random variable such that 
	$$\sup_{\mathbf{Q} \in \mathcal{Q}} \mathbf{Q}(\mathbf{f}) < \infty.$$ 
	Then the superhedging duality
	$$ \pi(\mathbf{f}) =  \sup_{\mathbf{Q} \in \mathcal{Q}} \mathbf{Q}(\mathbf{f}) $$
	 holds and there exists some superhedging strategy $H \in \mathcal{A}$ such that 
	$$\pi(\mathbf{f}) + H \cdot \mathbf{S}_T \ge \mathbf{f}, a.s.$$ 
\end{theorem}
A contingent claim $\mathbf{f}$ is a vector of random variables in $\mathbf{L}^0(\mathcal{F}_T,P)$ such that $\mathbf{f} \ge 0, a.s.$. A contingent claim $\mathbf{f}$ is called replicable if there exist $x \in \mathbb{R}$ and $H \in \mathcal{A}$ such that 
$$x + H \cdot \mathbf{S}_T = \mathbf{f}, a.s..$$ The market is complete if all contingent claims are replicable. Our third main result is the following. 
\begin{theorem}\label{thm:complete}
	Under $NRA(\Theta)$, the market is complete if and only if $|\mathcal{Q}| = 1$. Furthermore, if $NA(\{\theta\})$ holds for some $\theta \in \Theta$, then $\Theta= \{ \theta \}$.
\end{theorem}

\subsection{The case with options}
Usually, information from the market data, for example available option prices, are used in calibration to select models that fit real data. We show in this section that this fitting procedure helps to reduce the set of robust pricing systems, and hence superhedging prices.

Let $e \in \mathbb{N}$ and $g^i: \Omega \to \mathbb{R}^d, i =1,...,e$ be traded options which can be only bought or sold at time $t = 0$ at the market price $g^i_0$. We may assume that $g^i_0=0, i = 1,...,e$ and the option $g^i$ will give the payoff $\mathbf{g}^i := (g^i(S^{\theta}_T))_{\theta \in \Theta}$ at time $T$. For a vector $\mathbf{a} = (a_1,...,a_e) \in \mathbb{R}^e$, the option portfolio from $\mathbf{a}$ is given by $\sum_{i=1}^e a^i \mathbf{g}^i$.   
A semi-static strategy $(H, \mathbf{a})$ is a pair of $H \in \mathcal{A}$ and $\mathbf{a} \in \mathbb{R}^e$, and the corresponding wealth at time $T$ is
$$ H \cdot \mathbf{S}_T + \sum_{i=1}^e a^i \mathbf{g}^i  = \sum_{s = 1}^T H_s \Delta \mathbf{S}_s + \sum_{i=1}^e a^i\mathbf{g}^{i}.$$
\begin{definition}\label{defi:NRA_O}
	We say that the market satisfies the condition $NRA(\Theta)$ if for every self-financing strategy $H \in \mathcal{A}$ and for every $ \mathbf{a} \in \mathbb{R}^e$ such that
	\begin{eqnarray*}
		&&\forall \theta \in \Theta,  H \cdot \mathbf{S}_T + \sum_{i=1}^e a^i \mathbf{g}^i \ge 0, a.s. \nonumber \\
		&&\qquad \text{ then } \qquad \forall \theta \in \Theta,  H \cdot \mathbf{S}_T + \sum_{i=1}^e a^i \mathbf{g}^i  = 0, a.s.. 
	\end{eqnarray*}
\end{definition}
\begin{definition}
	A calibrated robust pricing system is a robust pricing system which is consistent with the option prices. We define
	$$\mathcal{Q}^{\theta}_{cal,e}= \{ \mathbf{Q} \in \mathcal{Q}^{\theta}: \mathbf{Q}(g^i(\mathbf{S}_T)) = 0, i = 1,...,e \}$$
	and $\mathcal{Q}_{cal,e} = \cup_{\theta \in \Theta} \mathcal{Q}^{\theta}_{cal,e}$. 
\end{definition}
\begin{theorem}\label{thm:ftap_o}
	Under the setting above
	\begin{itemize}
		\item[(a)] The following are equivalent:
		\subitem(i) $NRA(\Theta)$ holds.
		\subitem(ii) $ \forall \theta \in \Theta, \mathcal{Q}^{\theta}_{cal,e} \ne \emptyset$.
		\item[(b)] Let $NRA(\Theta)$ hold, and $\mathbf{f}$ be a random variable. The superhedging price is defined by
		$$\pi^e(\mathbf{f}):= \inf\{ x \in \mathbb{R}: \exists (H, \mathbf{a}) \in \mathcal{A} \times \mathbb{R}^e \text{ such that } x + H \cdot \mathbf{S}_T + \sum_{i=1}^e a^i \mathbf{g}^i \ge \mathbf{f}, a.s. \}.$$
		Then the superhedging duality holds
		$$ \pi^e(\mathbf{f}) = \sup_{\theta \in \Theta} \sup_{Q \in \mathcal{Q}^{\theta}_{cal,e}} \mathbf{Q}(\mathbf{f}),$$
		and there exits $(H, \mathbf{a}) \in \mathcal{A} \times \mathbb{R}^e$ such that  $$\pi^e(\mathbf{f}) + H \cdot \mathbf{S}_T + \sum_{i=1}^e a^i g^i(\mathbf{S}_T) \ge \mathbf{f}, a.s.$$ 
		\item[(c)] Let $NRA(\Theta)$ hold, and $\mathbf{f}$ be a random variable. The following are equivalent
		\subitem(i) $\mathbf{f}$ is replicable.
		\subitem(ii) The mapping $\mathbf{Q} \mapsto \mathbf{Q}(\mathbf{f})$ is constant on $\mathcal{Q}_{cal,e}.$  
	\end{itemize}
\end{theorem}
The proof of this theorem is given in Subsection \ref{proof:thm_ftap_o}. 
\subsection{Examples}

\subsubsection{Drift and volatility uncertainty}\label{ex:1}
We consider Example \ref{ex:toy} with one-period, i.e., $T = 1$, and $|\Theta| =2$. The dynamics of the risky asset {\color{red}is} given by 
$$ S^{\theta_i}= 1 + \sigma_i \omega + \mu_i$$
where $\mu_i \in \mathbb{R}, 0 < \sigma_i$ for $i \in \{1,2\}$. If there are no confusion, we simply write $i$ instead of $\theta_i$. The requirements for a robust pricing system $\mathbf{Q}$ are
$$\left\langle \mathbf{Q}, (H\Delta S^{i} - h^{i} )_{i \in \Theta} \right\rangle \le 0, \qquad \forall H \in \mathbb{R}, h^{i} \in \mathbb{R}_+,$$
and the normalizing condition $\mathbf{Q}(\mathbf{1}) = 1$. In this example, $\mathbf{Q} = (Z^{1}, Z^{2})$ and further computation lead to the following system of equations
\begin{eqnarray}
	Z^1(1)[\sigma_1 + \mu_1] + Z^2(1)[\sigma_2 + \mu_2] + Z^1(-1)[-\sigma_1 + \mu_1] + Z^2(-1)[-\sigma_2 + \mu_2] &=& 0, \nonumber\\
	Z^1(1) + Z^2(1) + Z^1(-1)+ Z^2(-1) &=& 1,	\label{eq:sys0}
\end{eqnarray}
where $Z^{i}(j)$ denotes the value of  $Z^i$ when $\omega = j$. This system admits a solution if 
\begin{eqnarray}
	\min\{ \sigma_1 + \mu_1, \sigma_2 + \mu_2,  -\sigma_1 + \mu_1, -\sigma_2 + \mu_2 \} &<& 0, \label{eq:1} \\
	\max\{ \sigma_1 + \mu_1, \sigma_2 + \mu_2,  -\sigma_1 + \mu_1, -\sigma_2 + \mu_2 \} &>& 0. \label{eq:2}
\end{eqnarray}
We consider the following particular cases.
\begin{enumerate}
	\item The case $\mu_1 > \sigma_1, 0 > -\sigma_2 > \mu_2$. In this case, the first dynamics increases while the second one decreases, and each of them admits arbitrage opportunities. However, there is no robust arbitrage. Indeed, if $H$ is a robust arbitrage then it should satisfy the following conditions
	$$H (\sigma_1 \omega + \mu_1) \ge 0, \qquad H (\sigma_2 \omega + \mu_2) \ge 0, \forall \omega \in \Omega, $$ 
	which imply that $H = 0.$ In other words, the condition $NRA(\Theta)$ holds true.
	\item The case $|\mu_i| < \sigma_i$ for each $i = 1,2$. As shown in \cite{rasonyi2018utility}, for $i=1,2$, the martingale measure
		\begin{equation}\label{eq:EMM}
			Q^{\theta_i}(\omega) = \frac{1}{2}\left( 1 - \omega \frac{\mu_i}{\sigma_i} \right)
		\end{equation}
		is unique for  $S^{\theta_i}.$
\end{enumerate}
We find the solutions to \eqref{eq:sys0}. Let $0 \le \alpha \le 1$ and $\beta \in \mathbb{R}$ be chosen later. Consider
	\begin{eqnarray}
		Z^1(1)[\sigma_1 + \mu_1] - Z^1(-1)[\sigma_1 - \mu_1] = \beta, \nonumber\\
		Z^1(1) + Z^1(-1) = \alpha, \label{eq:sys1}
	\end{eqnarray}
	and
	\begin{eqnarray}
		Z^2(1)[\sigma_2 + \mu_2] - Z^2(-1)[\sigma_2 - \mu_2] = -\beta, \nonumber\\
		Z^2(1) + Z^2(-1) = 1- \alpha. \label{eq:sys2}
	\end{eqnarray}
	Now we can solve explicitly
	\begin{eqnarray}
		Z^1(1) &=& \frac{1}{2 \sigma_1}\left[ \beta + \alpha(\sigma_1 - \mu_1) \right], \nonumber \\
		Z^1(-1) &=& \frac{1}{2 \sigma_1}\left[ -\beta + \alpha(\sigma_1 + \mu_1) \right], \nonumber\\
		Z^2(1) &=& \frac{1}{2 \sigma_2}\left[ -\beta + (1-\alpha)(\sigma_2 - \mu_2) \right], \nonumber\\
		Z^2(-1) &=& \frac{1}{2 \sigma_2}\left[ \beta + (1-\alpha)(\sigma_2 + \mu_2) \right].\label{eq:explicit_Q}
	\end{eqnarray}
	For each pair $(\alpha, \beta)$ such that
	\begin{eqnarray}
		0 &\le& \beta + \alpha (\sigma_1 - \mu_1) \le 2 \sigma_1, \nonumber\\
		0 &\le& -\beta + \alpha (\sigma_1 + \mu_1) \le 2 \sigma_1, \nonumber\\
		0 &\le& -\beta + (1-\alpha) (\sigma_2 - \mu_2) \le 2 \sigma_2, \nonumber\\
		0 &\le& \beta + (1-\alpha) (\sigma_2 + \mu_2) \le 2 \sigma_2  \label{eq:conditions}
	\end{eqnarray}
	we can construct a robust pricing system. The choice $\alpha = 1, \beta =0$ corresponds to the robust pricing system $(Q^{\theta_1},0)$ and the choice $\alpha = 0, \beta =0$ leads to $(0,Q^{\theta_2})$, where $Q^{\theta_i}, i=1,2$ are given in (\ref{eq:EMM}). Note that there are robust pricing systems different from linear combinations of $(Q^{\theta_1},0)$ and $(0,Q^{\theta_2})$.
	
	The parameters $\alpha$ and $(1-\alpha)$ are the weights put on the model $\theta_1$ and $\theta_2$, respectively. The parameter  $\beta$ controls the average of all outcomes under $\mathbf{Q}^1$ which can be positive (in the classical setting without uncertainty, this should be zero). However, this additional gain for the model $\theta_1$ is exactly compensated by an opposite gain $-\beta$ in the model $\theta_2$ so that there is no positive gain on average of all models.

In this example, we can compare the no robust arbitrage condition to other conditions introduced in Definition \ref{defi:swNA}. Under $NRA(\Theta)$, it may happen that $sNA(\Theta)$ and $wNA(\Theta)$ fail (the first case). The set $\{\text{Law}(S^{\theta_i}), i=1,2\}$ is not convex. However, using Lemma 3.2 of \cite{blanchard2019no}, we can always work with its convex hull $\mathcal{P} = Conv\{\text{Law}(S^{\theta_i}), i=1,2\}$, which is dominated by $\lambda \text{Law}(S^{\theta_1}) + (1-\lambda)\text{Law}(S^{\theta_2})$ for $\lambda \in [0,1]$. Therefore, $sNA(\mathcal{P})$ fails but $wNA(\mathcal{P})$ and $NA(\mathcal{P})$ holds true, see Theorem 3.6 of \cite{blanchard2019no}. The First Fundamental Theorem of \cite{bn15} is applied to the convex hull $\mathcal{P}$ and it can be checked that for any $\lambda \in (0,1)$ the canonical process $\widehat{S}$ satisfies the classical no arbitrage condition under the measure $\lambda\text{Law}(S^{\theta_1}) + (1-\lambda)\text{Law}(S^{\theta_2})$  and hence there exists a martingale measure $\widehat{Q}^{\lambda}$ for $\widehat{S}$. Compared to the results of  \cite{bn15}, a robust pricing system in the present paper can be interpreted as a martingale measure which is equivalent to a finite mixture of $\text{Law}(S^{\theta}), \theta \in \Theta$.

\subsubsection{Worst-case superhedging prices are not enough}\label{subsec:worstprice}
In the example of Subsection \ref{ex:1}, robust superhedging prices can be computed explicitly. Consider $\mu_1 = \mu_2 = 0$, $0 < \sigma_1 < \sigma_2 < 1$ and thus $P$ is the unique martingale measure for each $S^{i}, i = 1,2$. 

Consider a digital option $\mathbf{f} = (f^1,f^2)$ where $f^i = 1_{S^i_1 < 1}$. It can be checked that 
	$$\left(x^1, H^1\right) = \left( \frac{ 1}{2}, -\frac{1}{2\sigma_1} \right), \qquad \left(x^2, H^2\right) = \left( \frac{1}{2}, -\frac{1}{2\sigma_2} \right) $$
	are the superhedging prices and superhedging strategies of $f^1, f^2$ for the models $\theta_1, \theta_2$, respectively. 
	Now, we consider the robust superhedging problem for the claim $\mathbf{f}$, that is to solve the linear programming problem 
	\begin{eqnarray*}
		\min x, \text{ such that } && \\
		x + H \sigma_1 &\ge& 0,\\
		x - H \sigma_1 &\ge& 1,\\
		x + H \sigma_2 &\ge& 0,\\
		x - H \sigma_2 &\ge& 1.
	\end{eqnarray*}  
	This problem admits the solution $x= \frac{\sigma_2}{\sigma_2 + \sigma_1}$ and the corresponding strategy $H = -\frac{1}{\sigma_2 + \sigma_1}$. It is noticed that the worst-case superhedging price is $\max\{x^1,x^2\} = \frac{1}{2}$ 
	but there is no strategy $H$ that superhedges the claim $\mathbf{f}$ with this price, otherwise $H$ would be $H^1$ and $H^2$ simultaneously. This shows that the worst-case price is not enough for superhedging under uncertainty.
	
	The robust pricing systems are used to compute the robust superhedging price of $\mathbf{f}$ as follows
	\begin{eqnarray*}
		&&\max_{\mathbf{Q} \text{ satisfies } (\ref{eq:sys0})} \left( Z^1(1) \times 0 + Z^1(-1) \times 1 + Z^2(1) \times 0 + Z^2(-1) \times 1 \right) \\
		&=&\max_{\mathbf{Q} \text{ satisfies } (\ref{eq:sys0})} \left( Z^1(-1) + Z^2(-1)\right) \\
		%	&=&\max_{\mathbf{Q} \text{ satisfies } (\ref{eq:sys0})} 1 + (\sigma_2 - \sigma_1)\left( Z^1(1)+ Z^2(-1)\right)\\
		&=& \frac{\sigma_2}{\sigma_2 + \sigma_1}.
	\end{eqnarray*}
	The maximum is attained when $Z^1(1) = 0, Z^1(-1) = \frac{\sigma_2}{\sigma_1 + \sigma_2}, Z^2(1) = \frac{\sigma_1}{\sigma_1 + \sigma_2}, Z^2(-1)  = 0$. This simple example shows that it is not necessary to require $Z^{\theta}_T >0, a.s.,$ a property obtained by the exhaustion argument.

\section{Proofs}\label{sec:main}
\subsection{Preliminary results}
In this subsection, we closely follow the predictable range approach, given in \cite{delbaen2006mathematics}. The idea will be explained after introducing necessary notations.
\begin{lemma}\label{lemma:projection}
	Let $(\Omega, \mathcal{F}, P)$ be a probability space and $\mathcal{E} \subset L^0(\mathcal{F},P, \mathbb{R}^d)$ a subspace closed with respect to convergence in probability. We assume that $\mathcal{E}$ satisfies the following stability property. If $f,g \in \mathcal{E}$ and $A \in \mathcal{F},$ then $f1_A + g1_{A^c} \in \mathcal{E}$. Under these assumptions, there exists an $\mathcal{F}$-measurable mapping $\mathfrak{P}$ taking value in the orthogonal projection in $\mathbb{R}^d$ such that $f \in \mathcal{E}$ if and only if $\mathfrak{P}f = f$.
\end{lemma}
\begin{proof}
	See Lemma 6.2.1 of \cite{delbaen2006mathematics}.
\end{proof}
We define the following closed subspaces of $L^0(\mathcal{F}_{t-1},P), 1\le t \le T,$
$$\mathcal{E}^{\theta}_{t} = \{ h \in L^0(\mathcal{F}_{t-1},P): h \Delta S^{\theta}_t = 0, a.s.\}$$
and $\mathcal{E}^{\Gamma}_t = \bigcap_{\theta \in \Gamma} \mathcal{E}^{\theta}_t$ for $\Gamma \subset \Theta$. Each $\mathcal{E}^{\theta}_{t}$ satisfies the assumptions in Lemma \ref{lemma:projection} and so $\mathcal{E}^{\Gamma}_t$ does. Note that $0 \in \mathcal{E}^{\Gamma}_t$.  By Lemma \ref{lemma:projection}, $\mathcal{E}^{\Gamma}_t$ can be described by a mapping $\mathfrak{P}^{\Gamma}_t$.  We define $\mathfrak{P}^{\Gamma,c}_t = Id - \mathfrak{P}^{\Gamma}_t$ and 
$$\mathcal{H}^{\Gamma}_t = \{ f:\Omega \to \mathbb{R}^d: f \text{  is $\mathcal{F}_{t-1}$-measurable and  } \mathfrak{P}^{\Gamma,c}_tf = f \}.$$
By Assumption \ref{assum:Theta} (ii), if $\Gamma$ is a dense subset of $\Theta$ then $\mathcal{E}^{\Gamma}_t = \mathcal{E}^{\Theta}_t$ and $\mathcal{H}^{\Gamma}_t = \mathcal{H}^{\Theta}_t$. We say that $H \in \mathcal{A}$ is in $\Gamma$-canonical form if $H_t \in \mathcal{H}^{\Gamma}_t, 1\le t \le T$. And if $\Gamma$ is dense in $\Theta$, we simply say $H$ is in canonical form. Note that the random variable $H_t = 0$ is in $\mathcal{E}_t$. 

In classical settings without uncertainty, to prove the closedness of the set $\mathbf{C}^{\{\theta\}}$, we have to establish certain boundedness and convergence of a sequence of trading strategies $H_n$ from the same properties of  the corresponding hedgeable payoffs $f_n = H_{n} \cdot S^{\theta}_T$. This is not always possible since the sequence of trading strategies $H_n + n(h_1,...,h_T)$ for some $0 \ne h_t \in \mathcal{E}^{\theta}_t, t = 1,..,T,$ generates the same sequence of payoffs, however, does not have the required boundedness and convergence properties. Therefore, we have to restrict our analysis to non-redundant strategies, that are strategies in the canonical form. In Lemma \ref{lem:bounded_converge}, boundedness and convergence of trading strategies and the corresponding payoffs are given for one step models. Proposition \ref{pro:stricker} proves the closedness of the set $\mathbf{K}^{\Gamma}$, extending  Stricker's lemma. Proposition \ref{pro:crucial_bound} establishes the boundedness of trading strategies from the boundedness of terminal payoffs for multiple period models. We refer to \cite{delbaen2006mathematics} for further explanation of the predictable range approach.  
\begin{lemma}\label{lem:bounded_converge}
	Let $\Gamma$ be a dense subset of $\Theta$.	Let $(H_n)_{n \in \mathbb{N}} \in \mathcal{H}_t$ be a sequence in canonical form. It holds that
	\begin{itemize}
		\item[(i)] $(H_n)_{n \in \mathbb{N}}$ is a.s. bounded if and only if for all $\theta \in \Gamma$, $(H_n \Delta S^{\theta}_{t})_{n \in \mathbb{N}}$ is a.s. bounded.
		\item[(i')] Assume in addition that $NRA(\Theta)$ holds. Then $(H_n)_{n \in \mathbb{N}}$ is a.s. bounded if and only if for all $\theta \in \Gamma$, $(H_n \Delta S^{\theta}_{t})^-_{n \in \mathbb{N}}$ is a.s. bounded.
		\item[(ii)] $(H_n)_{n \in \mathbb{N}}$ converges a.s. if and only if for all $\theta \in \Gamma$, $(H_n \Delta S^{\theta}_{t})_{n \in \mathbb{N}}$ does.
	\end{itemize}
\end{lemma}
\begin{proof}
	The ``only if" directions are obvious. It suffices to prove the ``if" directions.
	
	$(i):$ Assume that for each $\theta \in \Gamma$, $(H_n \Delta S^{\theta}_{t})_{n \in \mathbb{N}}$ is a.s. bounded. We prove $(H_n)_{n \in \mathbb{N}}$ is a.s. bounded, too. If this is not the case, by Proposition
	6.3.4 (i) of \cite{delbaen2006mathematics}, there is a measurably parameterised subsequence $(L_k)_{k \in \mathbb{N}} = (H_{\tau_k})_{k \in \mathbb{N}}$ such that $L_k$ diverges to $\infty$ on a set $B$ of positive measure. Note that $L_k, k \in \mathbb{N}$ are in canonical form. Let $\widehat{L}_k = \frac{L_k}{\|L_k\|} 1_{ B \cap \|L_k\| \ge 1}$. By passing to another measurably parameterised subsequence we may assume that $\widehat{L}_k \to \widehat{L}$, which is of canonical form and satisfies $\widehat{L} = 1$ on $B$. By assumption, it holds that $\widehat{L}_k \Delta S^{\theta}_{t} \to 0, a.s.$ for all $\theta \in \Gamma$. Consequently, $\widehat{L} \Delta S^{\theta}_{t} = 0,  a.s.$ for all $\theta \in \Gamma$, and thus for all $\theta \in \Theta$, by Assumption \ref{assum:Theta} (ii), which means that $\widehat{L}  \in \mathcal{E}^{\Theta}_t$. Therefore, $\widehat{L} \in \mathcal{E}^{\Theta}_t \cap \mathcal{H}^{\Theta}_t = \{ 0\}$, which is a contradiction.
	
	$(i'):$ We proceed as in $(i)$, noting that 
	$$(\widehat{L}_k \Delta S^{\theta}_{t})^- = \lim_{k \to \infty} (\widehat{L}_k \Delta S^{\theta}_{t})^- = 0, a.s., \qquad \forall \theta \in \Gamma$$
	and thus for all $\theta \in \Theta$, by Assumption \ref{assum:Theta} (ii). By $NRA(\Theta)$, we get that $\widehat{L} \Delta S^{\theta}_{t} = 0, a.s.$ for all $\theta \in \Theta$, which again implies $\widehat{L} = 0$, a contradiction.
	
	$(ii):$  We also prove by contradiction. Assume that $(H_n)_{n \in \mathbb{N}}$ does not converge a.s.. By $(i)$, we may assume that $(H_n)_{n \in \mathbb{N}}$ is a.s. bounded. Proposition 6.3.3 of \cite{delbaen2006mathematics} implies there is a measurably parameterised subsequence $(H_{\tau_k})_{k \in \mathbb{N}}$ converging to $H_0,$ a.s.. Applying Proposition 6.3.4 (ii) of \cite{delbaen2006mathematics} with $f_0 = H_0$, there is another measurably parameterised subsequence $(H_{\sigma_k})_{k \in \mathbb{N}}$ converging to $\widehat{H}_0,$ a.s. for which $P[H_0 \ne \widehat{H}_0] >0$. Note also that $H_0, \widehat{H}_0$ are in canonical form. We have
	$$(H_0 - \widehat{H}_0) \Delta S^{\theta}_{t} = \lim_{k \to \infty}  H_{\tau_k} \Delta S^{\theta}_{t} - \lim_{k \to \infty} H_{\sigma_k} \Delta S^{\theta}_{t} =0, a.s.,\qquad \forall \theta \in \Gamma,$$
	and hence for all $\theta \in \Theta$, by Assumption \ref{assum:Theta} (ii),
	which implies a contradiction.
\end{proof}
We extend Stricker's lemma, see \cite{stricker} and also \cite{schachermayer1992} for a proof, noting that the condition $NRA(\Theta)$ is not used here.
\begin{proposition}\label{pro:stricker}
	Let $\Gamma \subset \Theta$ be a countable index set. The vector space $$\mathbf{K}^{\Gamma} = \left\lbrace  \left( \sum_{t=1}^T H_t \Delta S^{\theta}_t\right)_{\theta \in \Gamma}, H \in \mathcal{A} \right\rbrace $$ is closed in $\prod_{\theta \in \Gamma}L^{0}(\mathcal{F}_T, P)$.
\end{proposition}
\begin{proof}
	The case $T=1$:  we use Lemma \ref{lem:subseq} in Appendix, subsection \ref{sec:app_prod} and then Lemma \ref{lem:bounded_converge} (ii). Let us suppose that assertion holds true for $T-1$, and fix the horizon $T$. By the inductive hypothesis, the set 
	$$\mathbf{K}^{\Gamma}_2 = \left\lbrace  \left( \sum_{t=2}^T H_t \Delta S^{\theta}_t\right)_{\theta \in \Gamma} , H \in \mathcal{A} \right\rbrace $$
	is closed in  $\prod_{\theta \in \Gamma}L^{0}(\mathcal{F}_T, P)$. 
	
	Let $\mathcal{H}_1$ be the set of strategies in canonical form defined as before. Let $I_1$ be the linear mapping 
	\begin{eqnarray*}
		\mathcal{H}_1 &\to& \prod_{\theta \in \Gamma}L^{0}(\mathcal{F}_T, P) \\
		H_1 &\mapsto& (H_1 \Delta S^{\theta}_1)_{\theta \in \Gamma}.
	\end{eqnarray*}
	Note that $I_1$ is continuous and injective. Indeed, let $H_1$ and $H'_1$ in $L^0(\mathcal{F}_0)$ be such that $H_1 \Delta S^{\theta}_1 = H'_1 \Delta S^{\theta}_1$ for all $\theta \in \Gamma$. Then $H_1 - H'_1 \in \mathcal{E}^{\theta}_1$ for all $\theta \in \Gamma$ and hence $H_1 - H'_1 \in \mathcal{E}^{\Gamma}_1$. Because $H_1 -H'_1 \in \mathcal{H}^{\Gamma}_1 $, we deduce that $H_1 -H'_1 \in \mathcal{H}^{\Gamma}_1 \cap \mathcal{E}^{\Gamma}_1 = \{0\}$. Therefore,  $I_1$ is injective.
	
	Let $F_1 = (I_1)^{-1}(\mathbf{K}^{\Gamma}_2 \cap I_1(\mathcal{H}_1))$ be a subset of $\mathcal{H}_1$. Since $\mathbf{K}^{\Gamma}_2$ is closed, the set $F_1$ is a closed subspace of $\mathcal{H}_1$. We prove that $F_1$ is stable in the sense of Lemma \ref{lemma:projection}. Let $H_1,G_1 \in F_1$ and $A \in \mathcal{F}_0$. We have that
		\begin{eqnarray*}
		f^{\theta} = H_1\Delta S^{\theta}_1  = \sum_{t=2}^T H_t \Delta S^{\theta}_t, \qquad  g^{\theta} = G_1\Delta S^{\theta}_1  = \sum_{t=2}^T G_t \Delta S^{\theta}_t,
		\end{eqnarray*}
	for some strategies $H, G \in \mathcal{A}$ and $\mathbf{f} = (f^{\theta})_{\theta \in \Gamma}, \mathbf{g} = (g^{\theta})_{\theta \in \Gamma} \in \mathbf{K}^{\Gamma}_2 \cap I_1(\mathcal{H}_1)$.  The strategy $(H_t1_A + G_t1_{A^c})_{t \in \{1,...,T\} } \in \mathcal{A}$ generates
	$$f^{\theta}1_A + g^{\theta}1_{A^c} = (H_1\Delta1_A + G_11_{A^c})  S^{\theta}_1 = \sum_{t=2}^T (H_t1_A + G_t1_{A^c}) \Delta S^{\theta}_t,$$
	which implies that $H_11_A + G_11_{A^c} \in F_1.$

	 By Lemma \ref{lemma:projection}, there is an $\mathcal{F}_0$-measurable mapping $\mathfrak{P}_0$ so that $f \in F_1$ if and only if $\mathfrak{P}_0 f = f$. Define
	\begin{equation}\label{eq:E1}
	E_1 = \{ H_1 \in \mathcal{H}_1: \mathfrak{P}_0H_1 = 0 \}.
	\end{equation}
	Since every $H_1 \in E_1 \subset  \mathcal{H}_1$ is of canonical form,  and the integral $(H_1 \Delta S^{\theta}_1)_{\theta \in \Gamma}$ with $H_1 \in E_1, H_1 \ne 0$ is not in $\mathbf{K}^{\Gamma}_2$, otherwise $H_1 \in F_1 \cap E_1 = \{0\}$. Furthermore,
	$$\mathbf{K}^{\Gamma} =\left\lbrace \left( \sum_{t=1}^T H_t \Delta S^{\theta}_t\right)_{\theta \in \Gamma}, H \in \mathcal{A}, H_1 \in E_1 \right\rbrace $$  
	and the decomposition of elements $\mathbf{f} \in \mathbf{K}^{\Gamma}$ into $\mathbf{f}= (H_1 \Delta S^{\theta}_1)_{\theta \in \Gamma} + \mathbf{f}_2, H_1 \in E_1, \mathbf{f}_2 \in \mathbf{K}^{\Gamma}_2$ is unique.
	
	Let $\mathbf{f}_n = (H_{n,1} \Delta S^{\theta}_1)_{\theta \in \Gamma} + \mathbf{f}_{2,n}, n \in \mathbb{N}$ be a sequence in $\mathbf{K}^{\Gamma}$ with $H_{n,1} \in E_1, \mathbf{f}_{2,n} \in \mathbf{K}^{\Gamma}_2$ such that $\mathbf{f}_n \to \mathbf{f}$ in $\prod_{\theta \in \Gamma}L^{0}(\mathcal{F}_T, P)$. We prove that $\mathbf{f} \in \mathbf{K}^{\Gamma}$. By Lemma \ref{lem:subseq} in Appendix, subsection \ref{sec:app_prod}, we find a subsequence, still denoted by $n$, such that $f^{\theta}_n \to f^{\theta}, a.s.$ for all $\theta \in \Gamma$. First we will show that $(H_{n,1})_{n \in \mathbb{N}}$ is a.s. bounded. Let $A = \{ \omega: \limsup_{n \to \infty} \|H_{n,1} \| = \infty \}$. By Proposition 6.3.4 of \cite{delbaen2006mathematics}, there exists an
	$\mathcal{F}_0$-measurable subsequence $(\tau_k)_{k \in \mathbb{N}}$ such that $H_{\tau_k,1} \to \infty$ on $A$. If $P[A] > 0$,  we apply Proposition 6.3.3 of \cite{delbaen2006mathematics} and assume that $\frac{H_{\tau_k,1}}{\|H_{\tau_k,1}\|} \to \psi_1, a.s.$ on the set $A$, where $\psi_1 = \psi_1 1_A \in E_1$, since $E_1$ is closed and stable in the sense of Lemma \ref{lemma:projection}. Clearly $\|\psi_1\| = 1$ on $A$. We have that for every $\theta \in \Gamma$,
	$$\left( \frac{H_{\tau_k,1}}{\|H_{\tau_k,1}\|} \Delta S^{\theta}_1\right)1_A + \frac{f^{\theta}_{2,\tau_k}}{\|H_{\tau_k,1}\|}1_A  \to 0, a.s. $$
	It follows that $\frac{f^{\theta}_{2,\tau_k}}{\|H_{\tau_k,1}\|}1_A \to - 1_A \psi_1 \Delta S^{\theta}_1$. By the closedness of $\mathbf{K}^{\Gamma}_2$, we obtain that $( - 1_A \psi_1 \Delta S^{\theta}_1)_{\theta \in \Gamma} \in \mathbf{K}^{\Gamma}_2$. Since $\psi_1 \in E_1$, this implies that $\psi_1 \Delta S^{\theta}_1 = 0$ for all $\theta \in \Theta$  and hence $\psi_1=0$. This is a contradiction to $\|\psi_1\| = 1$ on $A$ with $P[A >0]$.
	
	So $(H_{n,1})_{n \in \mathbb{N}}$ is bounded a.s. and there exists an
		$\mathcal{F}_0$-measurable sequence $H_{\tau_k,1}$ converging a.s. to $H_1$. This means that $\mathbf{f}_{2,\tau_k} \to \mathbf{f} - (H_1 \Delta S^{\theta}_1)_{\theta \in \Gamma}$. By the closedness of $\mathbf{K}^{\Gamma}_2$, we get $\mathbf{f} - (H_1 \Delta S^{\theta}_1)_{\theta \in \Gamma} \in \mathbf{K}^{\Gamma}_2$. The proof is complete.
\end{proof}
Next, we prove a crucial boundedness property by extending Lemma \ref{lem:bounded_converge} (i').
\begin{proposition}\label{pro:crucial_bound}
	Let $NRA(\Theta)$ hold and $\Gamma$ be a dense subset of $\Theta$. Let $(H_n)_{n \in \mathbb{N}}$ be a sequence of strategies in canonical form such that $(H_n \cdot S^{\theta}_T)^- $ is bounded a.s. for every $\theta \in \Gamma$. Then $(H_n)_{n \in \mathbb{N}} = (H_{n,1},...,H_{n,T})_{n \in \mathbb{N}}$ is bounded a.s..
\end{proposition}
\begin{proof}
The proof follows exactly as the proof of Proposition 6.9.1 (ii) of \cite{delbaen2006mathematics}, noting that the classical arbitrage is replaced by the $NRA(\Theta)$, using Assumption \ref{assum:Theta} (ii). 
\end{proof}
\subsection{The closedness of $\mathbf{C}^{\Theta}$}
First, we prove the following.
\begin{lemma}\label{lemma:close_dense}
	For any countable and dense subset $\Gamma \subset \Theta$, 
	$$
	\mathbf{C}^{\Gamma} = \left\lbrace  (H \cdot S^{\theta}_T - h^{\theta})_{\theta \in \Gamma}, H \in \mathcal{A}, h^{\theta} \in L^0_+(\mathcal{F}_T,P) \right\rbrace $$
	is closed in $\prod_{\theta \in \Gamma} L^0(\mathcal{F}_T,P)$.
\end{lemma}
\begin{proof}
Since $L^0(\mathcal{F}_{T},P)$ is a metric space, the product space $\prod_{\theta \in \Gamma} L^0(\mathcal{F}_{T},P)$ is metrizable, see Theorem 3.36 of \cite{guide2006infinite}.
%	 Proposition 9.3.9 of \cite{morris1989topology}. 
Let 
	$$f^{\theta}_n = H_n \cdot S^{\theta}_T - h^{\theta}_n \to f^{\theta}, \text{  in } L^0(\mathcal{F}_T,P) \qquad \theta \in \Gamma.$$
	We prove that $(f^{\theta})_{\theta \in \Gamma} \in \mathbf{C}^{\Gamma}$. By Lemma \ref{lem:subseq} in Appendix, subsection \ref{sec:app_prod}, we may assume that $f^{\theta}_n \to f^{\theta}, a.s., \forall \theta \in \Gamma.$  The proof continues exactly the same as the standard argument in \cite{kabanov2001teacher}.% under Assumption \ref{assum:Theta} (ii). 
\end{proof}
\begin{proposition}\label{pro:closedness}
	Let Assumption \ref{assum:Theta} be in force. Assume that the condition $NRA(\Theta)$ holds. Then the set $\mathbf{C}^{\Theta}$ is closed in $\mathbf{L}^0(\mathcal{F}_T,P) $ with respect to the product topology, that is if $(\mathbf{f}_{\alpha})_{\alpha \in I}$ be a net in $\mathbf{C}^{\Theta}$ and $\mathbf{f}_{\alpha} \to \mathbf{f}$ for some $\mathbf{f} \in \mathbf{L}^0(\mathcal{F}_T,P) $, then $\mathbf{f} \in \mathbf{C}^{\Theta}$.
\end{proposition}
\begin{proof}
	Let $\mathbf{f}_{\alpha}$ be a net in $\mathbf{C}^{\Theta}$, i.e., $\mathbf{f}_{\alpha} = H_{\alpha} \cdot  \mathbf{S}_{T} - \mathbf{h}_{\alpha}$ for some $H_{\alpha} \in \mathcal{A}, \mathbf{h}_{\alpha} \in \mathbf{L}^0_+(\mathcal{F}_{T},P)$ and $\mathbf{f}_{\alpha} \to \mathbf{f}$ in the product topology. For every $\theta \in \Theta$, 
	\begin{equation}\label{eq:con}
	H_{\alpha} \cdot S^{\theta}_{T} - h^{\theta}_{\alpha} \to f^{\theta},   \text{ in } L^0(\mathcal{F}_{T},P).
	\end{equation}
	We need to show that $\mathbf{f} = (H \cdot \mathbf{S})_T - \mathbf{h}$ for some $H \in \mathcal{A}$ and $\mathbf{h} \in \mathbf{L}^0_+(\mathcal{F}_T,P)$. 
	
	\noindent\textit{\underline{Step 1} (Finite intersection property)} Let $Fin(\Theta)$ be the set of all non-empty finite subsets of $\Theta$. Let $D \in Fin(\Theta)$ be arbitrary and denote
	$$\mathbb{H}_D = \{ H \in \mathcal{A}: H \text{  is in $D$-canonical form, } H \cdot S^{\theta}_{T} \ge f^{\theta},  a.s., \forall \theta \in D \}.$$
	It is easy to see that $\mathbb{H}_D $ is convex and closed with respect to the topology of convergence in probability. We will prove a finite intersection property, that is 
	\begin{equation}\label{eq:FIP}
	\mathbb{H}_D= \bigcap_{\theta \in D}\mathbb{H}_{\{\theta\}} \ne \emptyset.
	\end{equation}
	By Assumption \ref{assum:Theta} (i), there exists a sequence $(\theta_n)_{n \in \mathbb{N}} \subset \Theta$ which is dense in $\Theta$. We use the density of $\Gamma = (\theta_n)_{n \in \mathbb{N}} \cup D $ to prove a stronger statement
	\begin{equation}\label{eq:FIP2}
	\mathbb{H}_{\Gamma} \ne \emptyset.
	\end{equation}
	The set $\Gamma$ is also countable and the product space $\prod_{\Gamma} L^0(\mathcal{F}_{T},P)$ is metrizable. From (\ref{eq:con}) we can find a sequence $(\alpha_n)_{n \in \mathbb{N}} \subset I$, which will be denoted by $(n)_{n\in \mathbb{N}}$ without causing any confusion, such that
	$$f^{\theta}_n = H_{n} \cdot S^{\theta}_{T} - h^{\theta}_{n} \to f^{\theta}, \text{  in } L^0(\mathcal{F}_T,P)  \qquad  \forall \theta \in \Gamma.$$
	
	%We will prove (\ref{eq:FIP}) by using induction over the number of periods in the market. There is nothing to prove for $T=0$. Assume that the claim (\ref{eq:FIP}) holds true for any market with dates \{1,2,...,T\} and we will prove it for the case with dates $\{0,1,...,T\}$.
	By Lemma \ref{lemma:close_dense}, there exist $H \in \mathcal{A}, h^{\theta} \in L^0_+(\mathcal{F}_T,P)$ such that $f^{\theta} = H \cdot S^{\theta}_T - h^{\theta}$, for $\theta \in \Gamma$, or equivalently, (\ref{eq:FIP2}) holds true.
	
	\noindent \textit{\underline{Step 2} (Boundedness of $\mathbb{H}_{\Gamma}$)}
	We prove by contradiction that for each $t=1,...,T$, the convex set $\{ H_t, H \in \mathbb{H}_{\Gamma} \}$ is bounded in probability. If this is not the case, there are $\alpha > 0$ and a sequence $(H_n)_{n \in \mathbb{N}}$ in $\mathbb{H}_{\Gamma}$ such that for each $n \in \mathbb{N}$, 
	$$P[\|H_{n,t}\| \ge n] \ge \alpha - 1/n.$$
	Since $H_n \cdot S^{\theta}_T$ is bounded from below by $f^{\theta}$ for each $\theta \in \Gamma$, Proposition \ref{pro:crucial_bound} implies that $H_n$ is a.s. bounded, which is a contradiction.
	
	\noindent \textit{\underline{Step 3} (Convex compactness of $\mathbb{H}_{\Gamma}$)}  Let $I'$ be an arbitrary set and $(F_{i})_{i \in I'}$ a family of closed and convex subsets of $\mathbb{H}_{\Gamma}$. Assume that
	\begin{equation}\label{eq:fi}
	\forall D \in Fin(I'), \qquad G_D = \bigcap_{i \in D} F_i \ne \emptyset.
	\end{equation}
	We will prove that
	$$\bigcap_{i \in I'} F_i \ne \emptyset.$$
	Due to Equation (\ref{eq:fi}), for each $D \in Fin(I')$ we can choose $H_{D} \in G_D.$ Consider the net $(H^+_D - H^-_D)_{D \in Fin(I')}$. By Lemma 2.1 of \cite{pratelli}, for every $D \in Fin(I)$, there exists $\widetilde{H}^-_D \in conv\{ H^-_E, E \ge D \}$ such that the net $(\widetilde{H}^-_{D,t})_{D \in Fin(I')}$ converges in measure to a nonnegative real-valued random variable $H^-_t$, for $t \in \{1,...,T\}$. It should be emphasized from Step 2 that $conv \{ H^-_{E,t}, E \ge D \}$ is bounded in probability. Using the same weights as in the construction of $\widetilde{H}^-_D$, we obtain $\widetilde{H}^+_D$. Again, Lemma 2.1 of \cite{pratelli} and Step 2 imply that there exist $\widehat{H}^+_D \in  conv\{ \widetilde{H}^+_E, E \ge D \}$ such that $\widehat{H}^+_{D,t}$ converges to a nonnegative real-valued random variable $H^+_t$. Repeating this argument for each $t =1,...,T$ yields a strategy $(H_t)_{t=1,...,T}$. It is clear that $\widehat{H}_{D,t} \to H_t$ for $t \in \{1,...,T\}$, and $H \in \bigcap_{i \in I'} F_i$, which implies the desired convex compactness. 
	
	\noindent \textit{\underline{Step 4}} Finally, we have
	$$ \emptyset \ne \bigcap_{\theta \in \Theta} \left( \mathbb{H}_{\{\theta\}} \cap \mathbb{H}_{\Gamma}\right) \subset \bigcap_{\theta \in \Theta} \mathbb{H}_{\{\theta\}},$$
	and the proof is complete.
\end{proof}
\subsection{Proof of Theorem \ref{thm:robustftap}}
\begin{proof}[Proof of Theorem \ref{thm:robustftap}]
We may assume $S^{\theta}_t, \theta \in \Theta$ are in $L^1(\mathcal{F}_t,P)$, up to an equivalent measure change $\frac{dP_1}{dP} = c \exp( - sup_{\theta \in \Theta, t} \|S^{\theta}_t\| )$, where $c$ is a normalization constant.
	
	$(i) \implies (ii)$. By Proposition \ref{pro:closedness}, the set $\mathbf{C}^{\Theta}$ is closed in the product space $\mathbf{L}^0(\mathcal{F}_{T},P)$ and hence the convex set $\mathbf{C}^{\Theta} \cap \mathbf{L}^1(\mathcal{F}_{T},P)$  is closed in the product space $\mathbf{L}^1(\mathcal{F}_{T},P)$, too. Fix $\theta \in \Theta$ and $A \in \mathcal{F}_T$ with $P[A]>0$ arbitrarily. Since $NRA(\Theta)$ holds, the closed convex cone $\mathbf{C}^{\Theta} \cap \mathbf{L}^1(\mathcal{F}_{T},P)$ and the compact set $\{1_{A} \mathbf{1}^{\theta}\}$ are disjoint. The Hahn-Banach theorem implies that there is 
	$$\mathbf{Q} \in \left( \mathbf{L}^1(\mathcal{F}_{T},P)\right)^* = \bigoplus_{\theta \in \Theta}L^{\infty}(\mathcal{F}_{T},P),$$ 
	such that 
	\begin{equation}\label{eq:G}
		\sup_{H \in \mathcal{A}, \mathbf{h} \in \mathbf{L}^0_+(\mathcal{F}_{T},P)} \mathbf{Q}\left(H \cdot \mathbf{S}_{T} - \mathbf{h} \right) \le \alpha < \beta \le \mathbf{Q}(1_{A} \mathbf{1}^{\theta}) .
	%\mathbf{Q}\left(H \cdot \mathbf{S}_{T} - \mathbf{h} \right) \le \alpha, \qquad \forall H \in \mathcal{A}, \mathbf{h} \in \mathbf{L}^0_+(\mathcal{F}_{T},P)
	\end{equation}
We can identify $\mathbf{Q}$ with a finite vector of  continuous linear functions on $L^1(\mathcal{F}_{T},P)$, that is 
$$\mathbf{Q}(\mathbf{f}) = \sum^{}_{\theta' \in \Theta} E[Z^{\theta'}_Tf^{\theta'}]$$
where $Z^{\theta'}_T = 0$ for all but a finite number of $\theta'$. As $\mathbf{0} \in \mathbf{C}^{\Theta}$, we have that $\alpha \ge 0$ 	and $\mathbf{Q}(1_{A} \mathbf{1}^{\theta}) > 0$. Since $L^0(\mathcal{F}_{T},P)$ is a linear space, it necessarily holds that $\alpha = 0$,
	\begin{equation}\label{eq:hb}
	\mathbf{Q}\left( H \cdot \mathbf{S}_{T} \right) = 0,  \qquad \forall H \in \mathcal{A},
	\end{equation}
	and $Z^{\theta'}_T \ge 0$ for all $\theta' \in \Theta$ with $E[1_{A} Z^{\theta}_T>0] >0$. For any $A_{t-1} \in \mathcal{F}_{t-1}$, $1_{A_{t-1}}(\mathbf{S}_t - \mathbf{S}_{t-1}) \in \mathbf{C}^{\Theta}$, and (\ref{eq:hb}) implies that $\mathbf{Q}(1_{A_{t-1}}(\mathbf{S}_t - \mathbf{S}_{t-1}) ) = 0$, or equivalently, we obtain the generalized martingale property.

	$(ii) \implies (i)$. Assume there exists a robust arbitrage strategy $H \in \mathcal{A}$, that is $H \cdot S^{\theta'}_T \ge 0$ for all $\theta' \in \Theta$ and $P[H \cdot S^{\theta}_T > 0] >0$ for some $\theta.$ By $(ii)$, there exists $\mathbf{Q} \in \mathcal{Q}^{\theta}$ with $E[Z^{\theta}_T 1_{H \cdot S^{\theta}_T >0} ]>0,$ and thus 
	$$\mathbf{Q}(H \cdot \mathbf{S}_T) > 0.$$
	However, Lemma \ref{lem:mart_1} (iii) in Appendix, subsection \ref{sec:generalized} implies that $H \cdot \mathbf{S}_T$ is a generalized martingale under $\mathbf{Q}$ and then $\mathbf{Q}(H \cdot \mathbf{S}_T) = 0$, which is a contradiction.
\end{proof}	
%{\color{red}
%\begin{corollary}
% Under $NRA$, for every $\theta \in \Theta$, and for any $A \in \mathcal{F}_T$ with $P[A] > 0$, there exists $\mathbf{Q}\in \mathcal{Q}^{\theta}$ with $E[Z_T^{\theta} 1A] > 0.$
%\end{corollary}
%\begin{proof}%
%	DO WE NEED THIS?
%\end{proof}
%}
In order to obtain robust pricing systems for each $t$, we proceed as follows. We define 
\begin{equation}
\mathbf{Z}_t = (E[Z^{\theta}_T|\mathcal{F}_t])_{\theta \in \Theta}.
\end{equation}
Then $\mathbf{Z}_t\mathbf{S}_t$ is a martingale under $P$. Indeed, we compute for any $0 \le s \le t \le T, A_s \in \mathcal{F}_s$ that
\begin{eqnarray*}
	\sum_{\theta \in \Theta} E[(Z^{\theta}_t S^{\theta}_t - Z^{\theta}_s S^{\theta}_s )1_{A_s}] &=& \sum_{\theta \in \Theta} E[Z^{\theta}_t S^{\theta}_t1_{A_s}] - \sum_{\theta \in \Theta} E[Z^{\theta}_s S^{\theta}_s 1_{A_s}]\\
	&=& \sum_{\theta \in \Theta} E[Z^{\theta}_T S^{\theta}_t1_{A_s}] - \sum_{\theta \in \Theta} E[Z^{\theta}_T S^{\theta}_s 1_{A_s}] = 0.
\end{eqnarray*}
\begin{remark}
	For each $\theta \in \Theta, A \in \mathcal{F}_T $, it is only required that there exists a robust pricing system with  density $Z^{\theta}_T > 0$ on $A$. Here, we do not use the typical exhaustion argument to obtain the density $Z^{\theta}_T >0, P-a.s.$. The situation is discussed in more details in Subsection
	\ref{subsec:worstprice}. %If the model $\theta$ satisfies the classical no arbitrage condition, then the density  $Z^{\theta}_T$ of a martingale measure $Q^{\theta}$ for $S^{\theta}$ constitutes a robust pricing system $Z^{\theta}_T\mathbf{1}^{\theta}$, see Example \ref{ex:1}.
\end{remark}
\subsection{Proof of Theorem \ref{thm:superhedge}}
($\ge$) Let $x \in \mathbb{R}$ be such that $x + H \cdot  \mathbf{S}_T \ge \mathbf{f}, a.s.$ for some $H \in \mathcal{A}$. Then for an arbitrary $\mathbf{Q} \in \mathcal{Q}$, it holds that $x \ge \mathbf{Q}(\mathbf{f})$ by Lemma \ref{lem:mart_1} (iii) in Appendix, and therefore  $x \ge \sup_{\mathbf{Q} \in \mathcal{Q}} \mathbf{Q}(\mathbf{f})$.
	
\noindent ($\le$) Take $x < \pi(\mathbf{f})$. Consequently, we have $\mathbf{f} - x \mathbf{1}  \notin \mathbf{C}^{\Theta}.$ Since $\mathbf{C}^{\Theta}$ is closed, the Hahn-Banach theorem implies there exists $\mathbf{Q} \in \bigoplus_{\theta \in \Theta}L^{\infty}(\mathcal{F}_{T},P)$ such that for all $\mathbf{h} \in \mathbf{C}^{\Theta} $
	$$ \mathbf{Q}\left( \mathbf{h} \right) < \mathbf{Q}\left( \mathbf{f} - x \mathbf{1} \right).$$
	Since $\mathbf{C}^{\Theta}$ is a convex cone containing $-\mathbf{L}_+$, we have that
	$\mathbf{Q}\left( \mathbf{h} \right) \le 0$ for all $\mathbf{h} \in \mathbf{C}^{\Theta}$ and $\mathbf{Q}\left( \mathbf{f} - x \mathbf{1} \right) > 0.$ %Let $\mathbf{Q}' \in \mathcal{Q}$ and define $\mathbf{Q}'' = \alpha \mathbf{Q}' + (1-\alpha) \mathbf{Q} \in \mathcal{Q}$ for suitable $\alpha \in (0,1)$ such that $\mathbf{Q}''(\mathbf{f} - x \mathbf{1}) > 0$. 
Normalizing $\mathbf{Q}$ so that $\mathbf{Q}(\mathbf{1}) = 1$, we obtain that $x \le \sup_{\mathbf{Q} \in \mathcal{Q}} \mathbf{Q}(\mathbf{f})$.
\subsection{Proof of Theorem \ref{thm:complete}}

\begin{proposition}\label{pro:complete}
	Let $NRA(\Theta)$ hold and $\mathbf{f} \ge 0$ be a contingent claim. The following are equivalent
	\begin{itemize}
		\item[(i)] $\mathbf{f}$ is replicable.
		\item[(ii)] The mapping 
		\begin{eqnarray*}
			\mathcal{Q} &\to& \mathbb{R}\\
			\mathbf{Q} &\mapsto& \mathbf{Q}(\mathbf{f})
		\end{eqnarray*}  
		is constant.
	\end{itemize}
\end{proposition} 
\begin{proof}
	$(i) \implies (ii)$. Assume that there exist $x \in \mathbb{R}, H \in \mathcal{A}$ such that $x + H \cdot \mathbf{S}_T = \mathbf{f}, a.s..$ For any $\theta \in \Theta$ and $\mathbf{Q} \in \mathcal{Q}^{\theta}$, the process $\mathbf{S}$ is a generalized martingale under $\mathbf{Q}$ and thus $H \cdot \mathbf{S}_T$ is also a generalized martingale under $\mathbf{Q}$, by Lemma \ref{lem:mart_1} (iii) in Appendix, subsection \ref{sec:generalized}. The conclusion in $(ii)$ then follows by computing
	$$ \mathbf{Q}(\mathbf{f}) = \mathbf{Q}(x + H \cdot \mathbf{S}_T) = x,$$
	for every $\mathbf{Q} \in \mathcal{Q}$.
	
	$(ii) \implies (i)$. Let $H$ be a superhedging strategy for $\mathbf{f}$, that is 
	$\pi(\mathbf{f}) + H \cdot \mathbf{S}_T - \mathbf{f} \ge 0, a.s..$ If for some $\theta \in \Theta$, the inequality is strictly positive with strictly positive probability, then we get for all $\mathbf{Q} \in \mathcal{Q}^{\theta}$
	$$0 < \mathbf{Q}(\pi(\mathbf{f}) + H \cdot \mathbf{S}_T - \mathbf{f} ).$$
	Noting that $H \cdot \mathbf{S}_T$ is a generalized martingale under $\mathbf{Q}$ with zero expectation and using $(ii)$, we get
	$$ \sup_{\mathbf{Q} \in \mathcal{Q}^{\theta}} \mathbf{Q}(\mathbf{f}) = \sup_{\mathbf{Q} \in \mathcal{Q}} \mathbf{Q}(\mathbf{f}) < \pi(\mathbf{f}),$$ which contradicts Theorem \ref{thm:superhedge}. Therefore, $\mathbf{f}$ is replicable.
\end{proof}

\begin{proof}[Proof of Theorem \ref{thm:complete}]
	According to Proposition \ref{pro:complete}, we only need to prove the ``only if" part. Assume the market is complete. Then for every $A \in \mathcal{F}_T$ with $P[A] > 0$ and for every $\theta \in \Theta$, the claim $\mathbf{1}^{\theta}_{A}$ is replicable. Let $\mathbf{Q} = (Z^{\theta_1}_T,...,Z^{\theta_k}_T) \in \mathcal{Q}$ be fixed.  By Proposition \ref{pro:complete}, we obtain that
	\begin{equation}\label{eq:comp}
	\mathbf{Q}(\mathbf{1}^{\theta}_{A}) = \mathbf{Q'}(\mathbf{1}^{\theta}_{A})
	\end{equation}
	for every $\mathbf{Q}' = (Z'^{\theta}_T)_{\theta \in \Theta} \in \mathcal{Q}$. The choices $\theta = \theta_i, 1 \le i \le k$ give that $Z'^{\theta_i}_T = Z^{\theta_i}_T$. For $\theta \notin \{\theta_1,...,\theta_k\}$, we get that $Z'^{\theta}_T = 0$. This implies that there exists a unique robust pricing system $\mathbf{Q}$ and $|\Theta|$ is finite. 
	
	We consider the case $NA(\theta)$ holds for some $\theta \in \Theta$ and thus there is a martingale measure $Q^{\theta}$ for $S^{\theta}$ with the Radon-Nikodym density $Z^{\theta}_T$. Let $\mathbf{Q}^{\theta}= Z^{\theta}_T\mathbf{1}^{\theta}$ be the corresponding robust pricing system. If $\Theta$ is not a singleton, for every $\mathbf{Q}' \in \mathcal{Q}^{\theta'}$ where $ \theta' \ne \theta$ and $0 <\mathbf{Q}'(\mathbf{1}^{\theta'}_{A})$ for some $A \in \mathcal{F}_T$, we get from (\ref{eq:comp}) that 
	$$0 <\mathbf{Q}'(\mathbf{1}^{\theta'}_{A}) = \mathbf{Q}^{\theta}(\mathbf{1}^{\theta'}_{A}) = E[ Z_T^{\theta} \times 0] =  0.$$
	This contradiction shows that $\Theta = \{ \theta \}$ and the model $\theta$ is complete.
\end{proof}

%\begin{remark}If $|\Theta|$ is finite, a robust pricing system is the sum of at most $|\Theta|$ continuous linear functions on $L^1$, however, continuity is not guaranteed after the exhaustion argument. In general, a robust pricing system may involve a countable number of continuous linear function on $L^1(\mathcal{F}_{t+1},P)$, as a consequence of the exhaustion argument. We loose continuity, but not linearity. The example \ref{ex:2} below illustrate this point. Remark that there is a version of the Dominated convergence theorem.\end{remark}

\subsection{Proof of Theorem \ref{thm:ftap_o}}\label{proof:thm_ftap_o}
We repeat the arguments of \cite{bn15} for the sake of completeness, noting that unlike \cite{bn15}, we work with robust pricing systems and do not impose conditions on $\mathbf{f}$. For $e=0$, the results are true by Theorem \ref{thm:robustftap}, Theorem \ref{thm:superhedge}, and Proposition \ref{pro:complete}. Assume that the results hold true for the market with the stocks and $e \ge 0$ options. Now, we prove the theorem for the market with one more option $\mathbf{g}^{e+1}$ with the market price $g^{e+1}_0 = 0$. 
	
	Consider (a). Let $NRA(\Theta)$ hold. If $\mathbf{g}^{e+1}$ is replicable in the market consists of the stocks and available options $\mathbf{g}^1,...,\mathbf{g}^e$, we come back to the case with $e$ options and therefore we may assume that $\mathbf{g}^{e+1}$ is not replicable. If $g^{e+1}_0 \ge \pi^e(\mathbf{g}^{e+1})$, then we can construct a robust arbitrage by shorting one unit of $\mathbf{g}^{e+1}$ and using the initial capital $\pi^e(\mathbf{g}^{e+1})$ together with the superhedging strategy for $\mathbf{g}^{e+1}$, which exists by the induction hypothesis, to cover $\mathbf{g}^{e+1}$ at time $T$. Thus, the consistency with $RNA$ implies $g^{e+1}_0 < \pi^e(\mathbf{g}^{e+1})$. The induction hypothesis (b) gives
	$$ g^{e+1}_0 < \pi^e(\mathbf{g}^{e+1}) = \sup_{\theta \in \Theta}\sup_{\mathbf{Q} \in \mathcal{Q}^{\theta}_{cal,e}} \mathbf{Q}(\mathbf{g}^{e+1}),$$
	and since $\mathbf{g}^{e+1}$ is not replicable, by the induction hypothesis (c), there is $\theta_+ \in \Theta$,  $\mathbf{Q}^{\theta_+}_+ \in \mathcal{Q}^{\theta_+}_{cal,e}$ such that
	$$g^{e+1}_0 < \mathbf{Q}^{\theta_+}_+(\mathbf{g}^{e+1}) < \pi^e(\mathbf{g}^{e+1}).$$
	Similarly, we find $\theta_- \in \Theta$ and $\mathbf{Q}^{\theta_-}_- \in \mathcal{Q}^{\theta_-}_{cal,e} $ such that
	\begin{equation}\label{eq:sandwichmeasure}
		-\pi^e(-\mathbf{g}^{e+1})< \mathbf{Q}^{\theta_-}_-(-\mathbf{g}^{e+1}) < g^{e+1}_0 < \mathbf{Q}^{\theta_+}_+(\mathbf{g}^{e+1}) < \pi^e(\mathbf{g}^{e+1})
	\end{equation}
	By the induction hypothesis (a), for each $\theta \in \Theta$ there is $\mathbf{Q}^{\theta} \in \mathcal{Q}^{\theta}_{cal,e}$. Choosing appropriate weights $\lambda_-, \lambda_+, \lambda_0 \in (0,1)$ and $\lambda_- + \lambda_+ + \lambda_0 = 1$, we have that
	$$\mathbf{Q}' :=  \lambda_-\mathbf{Q}^{\theta_-}_- + \lambda_+\mathbf{Q}^{\theta_+}_+ + \lambda_0 \mathbf{Q}^{\theta} \in \mathcal{Q}^{\theta}_{cal,e} \text{ and } \mathbf{Q}'(\mathbf{g}^{e+1}) = 0.$$
	It means that $\mathbf{Q}' \in \mathcal{Q}^{\theta}_{cal,e+1}$. Thus we prove (i) implies (ii) in (a). The converse implication is easy. The proof of (c) is straightforward as well.
	
	Next we consider (b). Assume there are $x \in \mathbb{R}, H \in \mathcal{A}, \mathbf{a} \in \mathbb{R}^{e+1}$ such that $x + H \cdot \mathbf{S}_T + \sum_{i=1}^{e+1} a^i g^i \ge \mathbf{f}, a.s..$ We compute easily for every $\mathbf{Q} \in \mathcal{Q}_{cal,e+1}$ that $H \cdot \mathbf{S}$ is a generalized martingale under $\mathbf{Q}$ and $x = \mathbf{Q}(x + H \cdot \mathbf{S}_T + \sum_{i=1}^{e+1} a^i \mathbf{g}^i ) \ge \mathbf{Q}(\mathbf{f})$. As a result, we have
	\begin{equation}\label{eq:hedge_easy}
		\pi^{e+1}(\mathbf{f}) \ge \sup_{\theta \in \Theta} \sup_{\mathbf{Q} \in \mathcal{Q}^{\theta}_{cal,e+1}} \mathbf{Q}(\mathbf{f}).
	\end{equation}	
	Now we prove the reverse inequality,
	\begin{equation}\label{eq:hedge_reverse}
		\pi^{e+1}(\mathbf{f}) \le \sup_{\theta \in \Theta} \sup_{\mathbf{Q} \in \mathcal{Q}^{\theta}_{cal,e+1}} \mathbf{Q}(\mathbf{f}).
	\end{equation}
	We claim that 
	\begin{equation}\label{eq:claim}
		\text{there is a sequence $\mathbf{Q}_n \in \mathcal{Q}_{cal,e}$ such that $\mathbf{Q}_n(\mathbf{g}^{e+1}) \to 0, \mathbf{Q}_n(\mathbf{f}) \to \pi^{e+1}(\mathbf{f})$.}
	\end{equation} 
	Without loss of generality, we may assume that $\pi^{e+1}(\mathbf{f}) = 0$. If the claim (\ref{eq:claim}) fails, we get $0 \notin \overline{\{ (\mathbf{Q}(\mathbf{g}^{e+1}), \mathbf{Q} (\mathbf{f})),  \mathbf{Q} \in \mathcal{Q}_{cal,e} \}} \subset \mathbb{R}^2$. Using a separation argument, there are $\alpha, \beta \in \mathbb{R}$ such that
	\begin{equation}\label{eq:e}
		0 > \sup_{\mathbf{Q} \in \mathcal{Q}_{cal,e}} \mathbf{Q}(\alpha \mathbf{g}^{e+1} + \beta \mathbf{f}).
	\end{equation} 
	By (b) of the inductive hypothesis, it holds that
	$$ \sup_{\mathbf{Q} \in \mathcal{Q}_{cal,e}} \mathbf{Q}(\alpha \mathbf{g}^{e+1} + \beta \mathbf{f}) = \pi^e(\alpha \mathbf{g}^{e+1} + \beta \mathbf{f}).$$
	By definition $\pi^e(\mathbf{\psi}) \ge \pi^{e+1}(\mathbf{\psi})$ for any random variable $\mathbf{\psi}.$ Since $\mathbf{g}^{e+1}$ can be hedged at price $0$, we obtain $\pi^{e+1}(\alpha \mathbf{g}^{e+1} + \beta \mathbf{f}) = \pi^{e+1}(\beta \mathbf{f})$. Therefore,
	$$0 > \sup_{\mathbf{Q} \in \mathcal{Q}_{cal,e}} \mathbf{Q}(\alpha \mathbf{g}^{e+1} + \beta \mathbf{f}) \ge \pi^{e+1}(\beta \mathbf{f}).$$
	Clearly, $\beta \ne 0$. If $\beta > 0,$ we obtain $\pi^{e+1}(\mathbf{f}) < 0$, which contradicts our assumption that $\pi^{e+1}(\mathbf{f}) = 0$. Thus $\beta < 0$.  Since $\mathcal{Q}_{cal,e+1} \subset \mathcal{Q}_{cal,e}$, (\ref{eq:e}) implies that $0 > \mathbf{Q}'(\beta\mathbf{f})$, for $\mathbf{Q}' \in \mathcal{Q}_{cal,e+1}$. Consequently, $\mathbf{Q}'(\mathbf{f}) > 0 = \pi^{e+1}(\mathbf{f})$, contradicting to (\ref{eq:hedge_easy}). Therefore, the claim (\ref{eq:claim}) holds true. 
	
	Since $\mathbf{g}^{e+1}$ is not replicable, there are two robust pricing systems $\mathbf{Q}^{\theta_+}_+,\mathbf{Q}^{\theta_-}_-$ as in (\ref{eq:sandwichmeasure}). For the sequence $\mathbf{Q}_n$ as in (\ref{eq:claim}), we can find $\lambda^n_-, \lambda^n, \lambda^n_+ \in [0,1]$ such that $\lambda^n_- + \lambda^n + \lambda^n_+ = 1$ and 
	$$\mathbf{Q}_n'= \lambda^n_- \mathbf{Q}^{\theta_-}_- + \lambda^n \mathbf{Q}_n + \lambda^n_+ \mathbf{Q}^{\theta_+}_+ \in \mathcal{Q}_{cal,e}  \text{ satisfies } \mathbf{Q}_n'(\mathbf{g}^{e+1}) = 0,$$
	or equivalently, $\mathbf{Q}_n' \in \mathcal{Q}_{cal,e+1}$. By (\ref{eq:claim}), we can choose $\lambda^n_{\pm} \to 0$. Therefore, $\mathbf{Q}'_n(\mathbf{f}) \to 0$, which implies (\ref{eq:hedge_reverse}).

\section{Appendix}\label{sec:app}
\subsection{Product space}\label{sec:app_prod}
Let $I$ be an index set and for each $i \in I$, let $(X_i,\tau_i)$ be a topological space. The product space, denoted by $\prod_{i \in I}(X_i,\tau_i)$, consists of the product set $\prod_{i \in I}X_i$ and a topology $\tau$ having as its basis the family
$$\left\lbrace  \prod_{i \in I}O_i:  O_i \in \tau_i \text{ and } O_i = X_i \text{ for all but a finite number of } i \right\rbrace .$$
The topology $\tau$ is called the product topology, which is the coarsest topology for which all the projections are continuous. Note that the product space defined in this way is also a topological vector space, see Theorem 5.2 of \cite{guide2006infinite}. The direct sum $\bigoplus_{i \in I} X_i$  is defined to be the set of tuples $(x_i)_{i \in I}$ with $x_i \in X_i$ such that $x_i =0$ for all but finitely many $i$.

If each $(X_i, \tau_i)$ is locally convex then $\prod_{i \in I}X_i$ is locally convex, too. If $I$ is uncountable, the uncountable product space $\prod_{i \in I}X_i$  is not normable. Furthermore, the uncountable product space $\prod_{i \in I}X_i$ is not first countable, sequential closedness is different from topological closedness. The dual of the product space $\prod_{i \in I} X_i$ is algebraically equal to the direct sum of their duals $\bigoplus_{i \in I} X^*_i$, see Lemma 28.32 of \cite{Schechter1996handbook}. 

It is known that convergence in probability
implies almost sure convergence along a subsequence. The following lemma extends this result to countable products of $L^0$ spaces.  
\begin{lemma}\label{lem:subseq}
	Let $I$ be a countable set. If a sequence  $(\mathbf{f}_n)_{ n \in \mathbb{N}}$ in $\prod_{i \in I} L^0(\mathcal{F},P)$ converges to $\mathbf{f}$ in the product topology, then there exists a subsequence $(n_k)_{k \in \mathbb{N}}$ such that $f^{i}_{n_k} \to f^{i}, a.s.$ for all $i \in I$. 
\end{lemma}
\begin{proof}
	We may assume $I=\mathbb{N}$. The metric
	$$\mathbf{d}(\mathbf{f},\mathbf{g}) := \sum_{i \in I} \frac{1}{2^i} \frac{d(f^{i},g^{i})}{1+d(f^{i},g^{i})}$$
	induces the product topology on $\prod_{i \in I} L^0(\mathcal{F},P)$, where $d$ is the metric inducing the topology of convergence in probability. Let $(\varepsilon_k)_{k \in \mathbb{N}}$ be a sequence of positive numbers decreasing to zero. Since $\mathbf{f}_n \to \mathbf{f}$ in $\mathbf{d}$, there exists a subsequence $(n_k)_{k \in \mathbb{N}}$ such that for all $k \in \mathbb{N}$, we have
	$\mathbf{d}(\mathbf{f}_{n_k}, \mathbf{f}) < \frac{\varepsilon_k}{2^k},$
	and hence for every $i \in I$, $d(f^{i}_{n_k}, f^{i}) < \frac{2^i\varepsilon_k}{2^k}.$ From this we obtain 
	$$P[|f^{i}_{n_k} - f^{i}| \ge \varepsilon_k ] \le \frac{2^i}{2^k}.$$
	For every $\varepsilon >0$, there is $K$ such that $\varepsilon_k \le \varepsilon$ for $k \ge K$, and thus we compute that
	$$\sum_{k = K}^{\infty}P[|f^{i}_{n_k} - f^{i}| \ge \varepsilon ] \le  \sum_{k = K}^{\infty}P[|f^{i}_{n_k} - f^{i}| \ge \varepsilon_k ] \le  2^i$$
	for every $i \in I$. Therefore $\sum_{k = 1}^{\infty}P[|f^{i}_{n_k} - f^{i}| \ge \varepsilon ]$ is also finite. By the Borel–Cantelli lemma, for every $i \in I$, the set $$\{ \omega:  |f^{i}_{n_k}(\omega) - f^{i}(\omega)| \ge \varepsilon \text{ infinitely often } \}$$ has zero probability for all $\varepsilon >0$, or equivalently, $f^{i}_{n_k} \to f^{i}, a.s.$ for every $i \in I$. 
\end{proof}

\subsection{Convex compactness of $L^0_+$}
A set $A \subset L^0$ is bounded in probability if $\lim_{n \to \infty}\sup_{f \in A}P[|f| \ge n] =0.$
For any set $I$ we denote by $Fin(I)$ the family of all non-empty finite subsets of $I$. This is a directed set with respect to the partial order induced by inclusion.
We recall Definition 2.1 of \cite{gordan}.
\begin{definition}\label{defi:convex_compact}
	A convex subset $C$ of a topological vector space is \emph{convexly compact} if for any non-empty set $I$ and any family $(F_i)_{i \in I}$ of closed and convex subsets of $C$, the condition
	$$ \forall D \in Fin(I), \qquad \bigcap_{i \in D} F_i \ne \emptyset$$
	implies
	$$ \bigcap_{i \in I} F_i \ne \emptyset.$$
\end{definition}
The following result gives a characterization for convex compactness in $L^0_+$, see Theorem 3.1 of \cite{gordan}.
\begin{theorem}
	A closed and convex subset $C$ of $L^0_+$ is convexly compact if and only if it is bounded in probability.
\end{theorem}
\subsection{Generalized conditional expectation} \label{sec:generalized}
Let $\mathcal{F}$ be a sigma algebra. Let $I$ be an index set. In this subsection, we work with the product space $\mathbf{L}^1=\prod_{ i \in I} L^1(\mathcal{F}, P)$.  Let $\mathbf{Q}:\mathbf{L}^1 \to \mathbb{R}$ be a linear function such that $\mathbf{Q} = (Z^{i}_T)_{i \in I},$ $Z^{i}_T \ge 0$ for all $i \in I$ and $\mathbf{Q}(\mathbf{1}) = 1$.
\begin{definition}
	Let $\mathcal{G} \subset \mathcal{F}$ be two sigma algebras and $\mathbf{f} = (f_i)_{i \in I}$ be an $\mathcal{F}$-measurable random variable such that $\mathbf{Q}(\mathbf{f}) < \infty$. An $\mathcal{G}$-measurable random variable $\mathbf{f}_g$ is called a generalized conditional expectation of $\mathbf{f}$ with respect to $\mathcal{G}$ under $\mathbf{Q}$ if
	\begin{equation}\label{defi:gen_condi_expec}
	\mathbf{Q}(\mathbf{f}1_{A}) = \mathbf{Q}(\mathbf{f}_g1_{A}), \qquad \forall A \in \mathcal{G}.
	\end{equation} 
	We denote by $\mathbf{Q}(\mathbf{f}|\mathcal{G})$ the set of all generalized conditional expectations of $\mathbf{f}$. % For a general random variable $\mathbf{f}$ such that $\mathbf{Q}(\mathbf{f}^+ + \mathbf{f}^-) < \infty$, where $\mathbf{f}^+ = (f^+_i)_{i \in I}$ and $\mathbf{f}^- = (f^-_i)_{i \in I}$, we define $\mathbf{Q}(\mathbf{f}|\mathcal{G}): = \mathbf{Q}(\mathbf{f}^+|\mathcal{G}) - \mathbf{Q}(\mathbf{f}^-|\mathcal{G})?????????$.% and the convention that $\infty - \infty = - \infty$. 
\end{definition}
This definition becomes the usual definition for conditional expectation when $|I| = 1$. However, there are significant differences between the two concepts when $|I| \ge 2$. See Example \ref{ex:2} below for the facts that uniqueness and monotonicity fail in general.  
\begin{example}[Non-uniqueness]\label{ex:2}
	We consider the toy model with $T = 2$ and $|\Theta| = 2$, that is $\Omega = \{-1,1\}^2, \mathcal{F}_t = \Pi^{-1}_t(2^{\Omega_t})$ where $\Pi$ is defined in \eqref{eq:Pi} and
		$$ S^{\theta_i}_t = s_0 + \sum_{u=1}^t (\mu^i_{u} + \sigma^i_{u}\text{proj}_t(\omega) ), \qquad s_0 \in \mathbb{R}^d, \qquad i, t \in \{1,2\}.
		$$ %recall the setting of the toy model here, as people mightread the appendix before reaching the corresponding section in the paper.  We also assume that $|\mu_i| < \sigma_i, i = 1,2.$ A robust pricing system for $(S^{\theta_1}, S^{\theta_2})$ is given by
	$$\mathbf{Q}(\omega) = (Z^1_2,Z^2_2)(\omega) = \frac{1}{4} \left( \left( 1- \omega_1 \frac{\mu^1_1}{\sigma^1_1} \right) \left( 1- \omega_2 \frac{\mu^1_2}{\sigma^1_2} \right), \left( 1- \omega_1 \frac{\mu^2_1}{\sigma^2_1} \right) \left( 1- \omega_2 \frac{\mu^2_2}{\sigma^2_2} \right)   \right).$$ 
	Let us consider a conditional expectation of $(S^{\theta_1}_2,S^{\theta_2}_2)$ given $\mathcal{F}_1$, that is an $\mathcal{F}_1$-measurable vector $(X^1,X^2)$ such that
	\begin{eqnarray}
	E\left[ \left( Z^1_2S^{\theta_1}_2+ Z^2_2S^{\theta_2}_2 \right) 1_{\omega_1 = 1} \right] &=& E\left[\left( Z^1_2X^1 + Z^2_2X^2\right) 1_{\omega_1 = 1}\right] ,\label{eq:cond}\\
	E\left[\left( Z^1_2S^{\theta_1}_2+ Z^2_2S^{\theta_2}_2 \right) 1_{\omega_1 = -1} \right] &=& E\left[\left( Z^1_2X^1 + Z^2_2X^2\right) 1_{\omega_1 = -1}\right] .\label{eq:cond2}  
	\end{eqnarray} 
	It is easy to see that $(S^{\theta_1}_1,S^{\theta_2}_1)$ is a solution to the system of equations (\ref{eq:cond}), (\ref{eq:cond2}) since $Z^1_2, Z^2_2$ are classical martingale densities for $S^{\theta_1}, S^{\theta_2}$, respectively. Other solutions are vectors $(S^{\theta_1}_1 + Y^1, S^{\theta_2}_1 + Y^2)$ where $(Y^1,Y^2)$ is $\mathcal{F}_1$-measurable such that
	\begin{eqnarray}
	E\left[\left( Z^1_2Y^1 + Z^2_2Y^2\right) 1_{\omega_1 = 1} \right] &=& 0, \label{eq:cond3}\\
	E\left[\left( Z^1_2Y^1 + Z^2_2Y^2\right) 1_{\omega_1 = -1}\right]  &=& 0. \label{eq:cond4}
	\end{eqnarray} 
	%For a random variable $Z$, let us write $Z(\omega_1,\omega_2)$ be the value of $Z$ on the event $\{\omega_1, \omega_2\}$.  Straightforward computation lead to\begin{eqnarray}
	%0 &=&  Q^1(1,1) Y^1(1,1) + Q^1(1,-1)Y^1(1,-1) + Q^2(1,1) Y^2(1,1)  \nonumber \\
	%   && \qquad + Q^2(1,-1)Y^2(1,-1)  , \label{eq:cond3}\\
	%0 &=& Q^1(-1,1) Y^1(-1,1) + Q^1(-1,-1)Y^1(-1,-1)  + Q^2(-1,1) Y^2(-1,1)\nonumber \\
	%  && \qquad + Q^2(-1,-1)Y^2(-1,-1). \label{eq:cond4}
	%\end{eqnarray} 
	In the case without uncertainty, for example $\Theta = \{ \theta_1\}$, there is the unique solution $Y^1(1) = Y^1(-1) =0$ to the system of equations (\ref{eq:cond3}), (\ref{eq:cond4}) and thus $S^{\theta_1}_1$ is the conditional expectation of $S^{\theta_1}_2$. %However, under uncertainty, solutions to (\ref{eq:cond}), (\ref{eq:cond2}) are vectors $(S^{\theta_1}_1 + Y^1, S^{\theta_2}_1 + Y^2)$ where $(Y^1,Y^2)$ satisfies (\ref{eq:cond3}), (\ref{eq:cond4}).
\end{example}
\begin{definition}\label{defi:mart}
	Let $(\Omega, \mathcal{F}, P)$ be a probability space, equipped with a filtration $(\mathcal{F}_t)_{t=0,...,T}$. 	An adapted process $(\mathbf{M}_t)_{t=0,...,T}$ is a generalized martingale under $\mathbf{Q}$ if
	\begin{itemize}
		\item[(i)] $E\left[ \left|  \sum_{i \in I} Z^{i}_t M^{i}_t \right|  \right] < \infty, $ for $0 \le t \le T$ ,
		\item[(ii)] $\mathbf{Q}(\mathbf{M}_t1_{A_s}) = \mathbf{Q}(\mathbf{M}_s1_{A_s}), \qquad \forall A_s \in \mathcal{F}_s, 0 \le s \le t \le T$,
	\end{itemize}  
	A process $\mathbf{M}$ is a generalized supermartingale under $\mathbf{Q}$ if (ii) is replaced by
	$$\mathbf{Q}(\mathbf{M}_t1_{A_s}) \le \mathbf{Q}(\mathbf{M}_s1_{A_s}), \qquad \forall A_s \in \mathcal{F}_s,  0 \le s \le t \le T$$
\end{definition}
It is easy to observe that if $\mathbf{M}$ is a generalized martingale under $\mathbf{Q}$ (resp. generalized supermartingale under $\mathbf{Q}$) then $\mathbf{Q}(\mathbf{M}_t) = \mathbf{Q}(\mathbf{M}_s)$ (resp. $\mathbf{Q}(\mathbf{M}_t) \le \mathbf{Q}(\mathbf{M}_s)$) for $s \le t$.
\begin{definition}
	An adapted process $(\mathbf{M}_t)_{t=0,...,T}$ is a generalized local martingale under $\mathbf{Q}$ if there
	is an increasing sequence $(\tau_n)_{n \in \mathbb{N}}$ of stopping times such that $P(\lim_{n \to \infty} \tau_n = \infty) = 1$ and that each stopped process $\mathbf{M}_{t \wedge \tau_n}1_{\tau_n >0}$ is a generalized martingale under $\mathbf{Q}$.
\end{definition}
The following facts are easily obtained, see Theorem 5.15  of \cite{schied} and \cite{jacod}.  
\begin{lemma}\label{lem:mart_1}
	Let $\mathbf{M} \ge \mathbf{0}$ be an adapted process with $\mathbf{M}_0 = \mathbf{1}.$ The following are equivalent:
	\begin{itemize}
		\item[(i)] $\mathbf{M}$ is a generalized martingale under $\mathbf{Q}$.
		\item[(ii)] If $H$ is predictable and bounded, then $\mathbf{V} = H \cdot \mathbf{M}$ is a generalized martingale under $\mathbf{Q}$.
		\item[(iii)] If H is predictable, then $\mathbf{V} = H \cdot \mathbf{M}$ is a generalized local martingale under $\mathbf{Q}$. If in addition that $E\left[ \left( \sum_{i \in I} Z^{i}_T V^{i}_T\right)^-  \right] < \infty$, then  $\mathbf{V}$ is a generalized martingale under $\mathbf{Q}$.
	\end{itemize}
\end{lemma}
\begin{proof}
	%Define $L_t:= \sum_{\theta \in \Theta} Z^{\theta}_t V^{\theta}_t$. It is easy to see that  $\textbf{V}$ is a generalized martingale under $\mathbf{Q}$ if and only if $L$ is a martingale under $P$. We follow the arguments in Theorem 5.15 of \cite{schied} to get the results.
	
	$(i) \implies (ii)$. Assume that $\sup_{0\le t \le T-1}|H_t| \le c$. It is easy to check that $\mathbf{V}$ satisfies the property (i) and (ii) in Definition \ref{defi:mart}. 
	
	$(ii) \implies (iii)$. Define $\tau_n:= \inf\{t:|H_{t+1}| > n\}$ and $H^n_t := H_t1_{t \le \tau_n}$ for $n \in \mathbb{N}$. Since $H$ is predictable and finite-valued, $(\tau_n)_{n \in \mathbb{N}}$ is a sequence of stopping times increasing a.s. to infinity. By (ii), or each $n \in \mathbb{N}$, the process $H^n \cdot \mathbf{M}$ is a generalized martingale under $\mathbf{Q}$. Noting that $\mathbf{V}_{t \wedge \tau_n} = H^n \cdot \mathbf{M}_t$, we obtain that $\mathbf{V}$ is a generalized local martingale under $\mathbf{Q}$.
	
	For convenient notations, we define $L_t := \sum_{i \in I} Z^{i}_t V^{i}_t$. By assumption, $E\left[ L^-_T \right] < \infty$. We show inductively that $E\left[ L^-_t \right] < \infty$ for $t=1,...,T-1$.  The generalized martingale property of $\mathbf{V}_{t \wedge \tau_n}1_{\tau_n >0}$ under $\mathbf{Q}$ implies that $\forall A_{t-1} \in \mathcal{F}_{t-1}$,
	\begin{equation}\label{eq:condi2}
	E\left[\sum_{i \in I} Z^{i}_T\left( V^{i}_{t \wedge  \tau_n} -V^{i}_{(t-1) \wedge  \tau_n}\right) 1_{\tau_n >0} 1_{A_{t-1}}  \right] = 0, 
	\end{equation}
	and consequently,
	\begin{eqnarray}
	E\left[ \sum_{i \in I} \left(  Z^{i}_tV^{i}_{t \wedge  \tau_n} -Z^{i}_{t-1}V^{i}_{(t-1) \wedge  \tau_n}\right) 1_{\tau_n >0} 1_{A_{t-1}} \right] &=& 0 \label{eq:condi3}\\
	E\left[ \left. \sum_{i \in I} \left( Z^{i}_t V^{i}_{t \wedge  \tau_n} - Z^{i}_{t-1} V^{i}_{(t-1) \wedge  \tau_n}\right) 1_{\tau_n > t-1}  \right|  \mathcal{F}_{t-1} \right] &=& 0, a.s.. \label{eq:condi}
	\end{eqnarray}
	Using the inequality $x^- \ge -x$, \eqref{eq:condi3} and (\ref{eq:condi}), we compute that
	\begin{eqnarray*}
		E\left[\left. \left(  \sum_{i \in I} Z^{i}_tV^{i}_t\right) ^-1_{\tau_n > t-1}\right|  \mathcal{F}_{t-1}\right]  &=& E\left[\left. \left(  \sum_{i \in I} Z^{i}_tV^{i}_{t \wedge \tau_n}\right) ^-1_{\tau_n > t-1}\right| \mathcal{F}_{t-1}\right]\\
		&\ge& - E\left[ \left.  \sum_{i \in I} Z^{i}_tV^{i}_{t \wedge \tau_n}1_{\tau_n > t-1} \right| \mathcal{F}_{t-1}\right]\\
		&=& - \sum_{i \in I} Z^{i}_{t-1}V^{i}_{t-1}1_{\tau_n > t-1}.
	\end{eqnarray*}
	Sending $n$ to infinity yields $E[L^-_t|\mathcal{F}_{t-1}] \ge L^-_{t-1},$ and hence, $E[L^-_{t-1}] \le E[L^-_t]$, which implies the assertion above. Next, the Fatou lemma and (\ref{eq:condi2}) give
	\begin{eqnarray*}
		E\left[ L^+_t \right]  &=& E\left[ \lim_{n \to \infty} \left(  \sum_{i \in I} Z^{i}_{t \wedge \tau_n}V^{i}_{t \wedge \tau_n}\right)^+ 1_{\tau_n >0}
		\right] \\
		&\le& \liminf_{n \to \infty} E\left[  \left(  \sum_{i \in I} Z^{i}_{t \wedge \tau_n}V^{i}_{t \wedge \tau_n}\right)^+ 1_{\tau_n >0}
		\right] \\
		&=& \liminf_{n \to \infty} E\left[  \left(  \sum_{i \in I} Z^{i}_{t \wedge \tau_n}V^{i}_{t \wedge \tau_n}\right) 1_{\tau_n >0} + \left(  \sum_{i \in I} Z^{i}_{t \wedge \tau_n}V^{i}_{t \wedge \tau_n}\right)^- 1_{\tau_n >0}
		\right] \\
		&<& \infty,
	\end{eqnarray*}
	since the process $\left( \sum_{i \in I} Z^{i}_{t \wedge \tau_n} V^{i}_{t \wedge \tau_n }\right)  1_{\tau_n >0}$ is also a martingale under $P$ and $\sum_{t=0}^T L^-_t$ is an integrable majorant for every $L^-_{t \wedge \tau_n}$. Therefore, $\mathbf{V}$ satisfies the condition (i) of Definition \ref{defi:mart}.  Noting that for all $n \in \mathbb{N}$
	$$ \left| \sum_{i \in I} Z^{i}_{t \wedge  \tau_n } V^{i}_{t \wedge \tau_n }\right|   1_{\tau_n >0} \le \sum_{t=0}^T |L_t| \in L^1(P),$$
	the dominated convergence theorem gives 
	\begin{eqnarray*}
		E\left[ \left.   \sum_{i \in I} Z^{i}_t V^{i}_t  \right|  \mathcal{F}_{t-1} \right] &=& E\left[ \left.  \lim_{n \to \infty}  \sum_{i \in I} Z^{i}_{t \wedge \tau_n} V^{i}_{t \wedge \tau_n } 1_{\tau_n >0}  \right|  \mathcal{F}_{t-1} \right]  \\
		&=& \lim_{n \to \infty} E\left[ \left.    \sum_{i \in I} Z^{i}_{t \wedge \tau_n} V^{i}_{t \wedge \tau_n }  1_{\tau_n >0} \right|  \mathcal{F}_{t-1}  \right] \\
		&=& \lim_{n \to \infty} Z^{i}_{(t-1) \wedge \tau_n} V^{i}_{(t-1) \wedge \tau_n} 1_{\tau_n >0}  \\
		&=&  \sum_{i \in I} Z^{i}_{t-1} V^{i}_{t-1}, a.s..
	\end{eqnarray*} 
	Therefore, $\mathbf{V}$ is a generalized martingale under $\mathbf{Q}$. 
	
	$(iii) \implies (i)$. Consider the strategy $H_s = 1$ with the corresponding wealth process  $\mathbf{V}_t = \mathbf{M}_t - \mathbf{M}_0$. From (iii) we obtain that $\mathbf{M}_t$ is a generalized martingale under $\mathbf{Q}$.
\end{proof}

%\begin{acknowledgements}
	%We thank Masaaki Fukasawa, Mikl\'os R\'asonyi, Martin Schweizer, the AE and the referees for constructive comments on the earlier versions of the paper. 
%\end{acknowledgements}


%\bibliography{robust}   % name your BibTeX data base

\begin{thebibliography}{}
	% etc
	
	\bibitem{acciaio2016model}
	Acciaio, B., Beiglböck, M., Penkner, F., Schachermayer, W. A model‐free version of the fundamental theorem of asset pricing and the super‐replication theorem. \emph{Mathematical Finance}, 26(2), 233-251 (2016).
	
	\bibitem{adot}	Aksamit, A., Deng, S., Obłój, J., Tan, X.  The robust pricing–hedging duality for American options in discrete time financial markets. \emph{Mathematical Finance}, 29(3), 861-897 (2019).
	
	\bibitem{avellaneda1995pricing}
	Avellaneda, M., Levy, A.,  Parás, A. Pricing and hedging derivative securities in markets with uncertain volatilities. \emph{Applied Mathematical Finance}, 2(2), 73-88 (1995).
	
	\bibitem{balbas2007infinitely}
	Balbás, A., Downarowicz, A.  Infinitely many securities and the fundamental theorem of asset pricing. \emph{Mediterranean Journal of Mathematics}, 4(3), 321-341 (2007).
	
	\bibitem{barlt} Bartl, D. Exponential utility maximization under model uncertainty for unbounded endowments. \emph{The Annals of Applied Probability}, 29(1), 577-612  (2019).
	
	\bibitem{bhz2015}
	Bayraktar, E., Huang, Y. J., Zhou, Z. On hedging American options under model uncertainty. \emph{SIAM Journal on Financial Mathematics}, 6(1), 425-447 (2015).
	
	\bibitem{bayraktar2017}
	Bayraktar, E., Zhou, Z. On arbitrage and duality under model uncertainty and portfolio constraints. \emph{Mathematical Finance}, 27(4), 988-1012 (2017).
	
	\bibitem{bz2017} Bayraktar, E., Zhou, Z. Super-hedging American options with semi-static trading strategies under model uncertainty. \emph{International Journal of Theoretical and Applied Finance}, 20(06), 1750036 (2017).
	
	\bibitem{bhp}
	Beiglböck, M., Henry-Labordere, P., Penkner, F. Model-independent bounds for option prices—a mass transport approach. \emph{Finance and Stochastics}, 17(3), 477-501 (2013).
	
	\bibitem{biagini2015robust}
	Biagini, S., Bouchard, B., Kardaras, C., Nutz, M.  Robust fundamental theorem for continuous processes. \emph{Mathematical Finance}, 27(4), 963-987 (2017).
	
	%\bibitem{bp}Biagini, S., Pınar, M. Ç. The robust Merton problem of an ambiguity averse investor. \emph{Mathematics and Financial Economics}, 11(1), 1-24 (2017)
	
	\bibitem{blanchard2019no}
	Blanchard, R., Carassus, L. No-arbitrage with multiple-priors in discrete time. \emph{Stochastic Processes and their Applications} (2020).
	
	\bibitem{bn15}
	Bouchard, B., Nutz, M. Arbitrage and duality in nondominated discrete-time models. \emph{The Annals of Applied Probability}, 25(2), 823-859 (2015).
	
	\bibitem{brown2001robust}
	Brown, H., Hobson, D., Rogers, L. C. Robust hedging of barrier options. \emph{Mathematical Finance}, 11(3), 285-314 (2001).
	
	\bibitem{burzoni2019pointwise}
	Burzoni, M., Frittelli, M., Hou, Z., Maggis, M., Obłój, J.  Pointwise arbitrage pricing theory in discrete time. \emph{Mathematics of Operations Research}, 44(3), 1034-1057 (2019).
	
	\bibitem{burzoni2016universal}
	Burzoni, M., Frittelli, M., Maggis, M. Universal arbitrage aggregator in discrete-time markets under uncertainty. \emph{Finance and Stochastics}, 20(1), 1-50 (2016).
	
	\bibitem{burzoni2017}
	Burzoni, M., Frittelli, M., \and Maggis, M.  Model-free superhedging duality. \emph{The Annals of Applied Probability}, 27(3), 1452-1477 (2017).
	
	\bibitem{carr2010}
	Carr, P., \and Lee, R. Hedging variance options on continuous semimartingales. \emph{Finance and Stochastics}, 14(2), 179-207 (2010).
	
	\bibitem{guide2006infinite}
	Charalambos, D., Aliprantis, B. Infinite Dimensional Analysis: A Hitchhiker's Guide. \emph{Springer} (2006).
	
	\bibitem{cfr2020}
	Chau, H. N., Fukasawa, M., Rásonyi, M. Super-replication with transaction costs under
	model uncertainty for continuous processes.
	{\em Mathematical Finance}, 32(4), 1066-1085 (2021).
	
	\bibitem{chau2019robust}
	Chau, H. N., Rásonyi, M. Robust utility maximisation in markets with transaction costs. \emph{Finance and Stochastics}, 23(3), 677-696 (2019).
	
	\bibitem{cox2011robust}
	Cox, A. M., Obłój, J. Robust pricing and hedging of double no-touch options. \emph{Finance and Stochastics}, 15(3), 573-605 (2011).
	

	
	\bibitem{dalang1990equivalent}
	Dalang, R. C., Morton, A., Willinger, W. Equivalent martingale measures and no-arbitrage in stochastic securities market models. \emph{Stochastics: An International Journal of Probability and Stochastic Processes}, 29(2), 185-201 (1990).
	
	\bibitem{davis2007}
	Davis, M. H., \and Hobson, D. G. The range of traded option prices. \emph{Mathematical Finance}, 17(1), 1-14 (2007).
	
	\bibitem{davis2014arbitrage}
	Davis, M., Obłój, J., Raval, V. Arbitrage bounds for prices of weighted variance swaps. \emph{Mathematical Finance}, 24(4), 821-854 (2014).
	
	\bibitem{delbaen1994general}
	Delbaen, F., Schachermayer, W. A general version of the fundamental theorem of asset pricing. \emph{Mathematische Annalen}, 300(1), 463-520 (1994).
	
	\bibitem{denis2006}
	Denis, L., \and Martini, C. A theoretical framework for the pricing of contingent claims in the presence of model uncertainty. \emph{The Annals of Applied Probability}, 16(2), 827-852 (2006).
	
	\bibitem{delbaen2006mathematics}
	Delbaen, F., Schachermayer, W. The mathematics of arbitrage. \emph{Springer Finance} (2006).
	
	\bibitem{deparis}
	Deparis, S., \and Martini, C. Superhedging strategies and balayage in discrete time. In \emph{Seminar on Stochastic Analysis, Random Fields and Applications IV, Birkhäuser, Basel}, pp. 205-219, (2004).
	
	\bibitem{dolinsky14}
	Dolinsky, Y., Soner, H. M. Robust hedging with proportional transaction costs. \emph{Finance and Stochastics}, 18(2), 327-347 (2014). 
	
	\bibitem{dolinsky14b}
	Dolinsky, Y., Soner, H. M. Martingale optimal transport and robust hedging in continuous time. \emph{Probability Theory and Related Fields}, 160(1), 391-427 (2014).
	
	\bibitem{dolinsky-soner2} Dolinsky, Y. and Soner, H. M. Convex Duality with Transaction Costs. \emph{Mathematics of Operations Research}, 42:448--471, (2017). 
	
	\bibitem{fahim2016}
	Fahim, A.,  Huang, Y. J. Model-independent superhedging under portfolio constraints. \emph{ Finance and Stochastics}, 20, 51-81 (2016).
	
	\bibitem{fajardo2017}
	Fajardo, S., Keisler, H. J. Model theory of stochastic processes (Vol. 14). 
	{\em Cambridge University Press}, (2017).
	
	\bibitem{schied}
	Föllmer, H., Schied, A. Stochastic finance: an introduction in discrete time. Walter de Gruyter (2002).
	
	\bibitem{halmossavage}
	Halmos, P. R., Savage, L. J. Application of the Radon-Nikodym theorem to the theory of sufficient statistics. \emph{The Annals of Mathematical Statistics}, 20(2), 225-241 (1949).
	
	\bibitem{hobson1998robust}
	Hobson, D. Robust hedging of the lookback option. \emph{Finance and Stochastics}, 2(4), 329-347 (1998).
	
	\bibitem{hobson2011}
	Hobson, D. The Skorokhod embedding problem and model-independent bounds for option prices. In \emph{Paris-Princeton Lectures on Mathematical Minance} p. 267-318, (2011). Springer, Berlin, Heidelberg.
	
	\bibitem{hobson2012}
	Hobson, D., \and Neuberger, A. Robust bounds for forward start options. \emph{Mathematical Finance: An International Journal of Mathematics, Statistics and Financial Economics}, 22(1), 31-56 (2012).
	
	\bibitem{hobson2016}
	Hobson, D., Neuberger, A. More on hedging American options under model uncertainty. arXiv:1604.02274 (2016).
	
	\bibitem{hobson2017}
	Hobson, D., Neuberger, A. Model uncertainty and the pricing of American options. \emph{Finance and Stochastics}, 21(1), 285-329 (2017).
	
	\bibitem{hou2018} Hou, Z., \and Obłój, J. Robust pricing–hedging dualities in continuous time. \emph{Finance and Stochastics}, 22(3), 511-567 (2018).
	
	\bibitem{hoover1984}
	Hoover, D. N., Keisler, H. J. 
	\newblock Adapted probability distributions.
	\newblock {\em Transactions of the American Mathematical Society}, 286(1), 159-201 (1984).
	
	\bibitem{jacod}
	Jacod, J., Shiryaev, A. N. Local martingales and the fundamental asset pricing theorems in the discrete-time case. \emph{Finance and Stochastics}, 2(3), 259-273 (1998).
	
	
	\bibitem{kabanov2001teacher}
	Kabanov, Y., Stricker, C. A teacher's note on no-arbitrage criteria. \emph{Séminaire de probabilités de Strasbourg}, 35, 149-152 (2001).
	
	\bibitem{Keisler2009}
	Keisler, H. J.,  Sun, Y. 
	\newblock Why saturated probability spaces are necessary. 
	\newblock {\em Advances in Mathematics}, 221(5), 1584-1607 (2009).
	
	\bibitem{khan2014}
	Khan, M. A.,  Zhang, Y. 
	\newblock On the existence of pure-strategy equilibria in games with private information: a complete characterization. 
	\newblock {\em Journal of Mathematical Economics}, 50, 197-202 (2014).
	
	\bibitem{klein2000fundamental}
	Klein, I.  A fundamental theorem of asset pricing for large financial markets. \emph{Mathematical Finance}, 10(4), 443-458 (2000).
	
	
	\bibitem{knight1921risk}
	Knight, F. H. Risk, uncertainty and profit. Boston, New York (1921).
	
	
	\bibitem{kreps1981arbitrage}
	Kreps, D. M. Arbitrage and equilibrium in economies with infinitely many commodities. \emph{Journal of Mathematical Economics}, 8(1), 15-35 (1981).
	
	\bibitem{li}
	Li, J., Phillips, P. C., Shi, S., Yu, J. 
	\newblock Weak identification of long memory with implications for inference. 
	\newblock {\em Available at SSRN 4140818}, (2022).
	
	\bibitem{lms}
	Liebrich, F. B., Maggis, M.,  Svindland, G. 
	\newblock Model uncertainty: A reverse approach. 
	\newblock {\em SIAM Journal on Financial Mathematics}, 13(3), 1230-1269 (2022).
	
	\bibitem{Lyons}
	Lyons, T. J. Uncertain volatility and the risk-free synthesis of derivatives. \emph{Applied Mathematical Finance}, 2(2), 117-133 (1995).
	
	

	\bibitem{nutz2014}
	Nutz, M.  Superreplication under model uncertainty in discrete time. \emph{Finance and Stochastics}, 18(4), 791-803 (2014).
	
	\bibitem{nutz2016}
	Nutz, M. Utility maximization under model uncertainty in discrete time. \emph{Mathematical Finance}, 26(2), 252-268 (2016).
	
	\bibitem{obloj2018unified}
	Obłój, J., Wiesel, J. A unified framework for robust modelling of financial markets in discrete time. \emph{Finance and Stochastics}, 25(3), 427-468 (2021).
	
	\bibitem{pratelli}
	Pratelli, M. A minimax theorem without compactness hypothesis. \emph{Mediterranean Journal of Mathematics}, 2(1), 103-112 (2005).
	
	\bibitem{quenez}
	Quenez, M. C. Optimal portfolio in a multiple-priors model. In \emph{Seminar on Stochastic Analysis, Random Fields and Applications IV (pp. 291-321)}. Birkhäuser, Basel (2004).
	
	\bibitem{rasonyi2003equivalent}
	Rásonyi, M. Equivalent martingale measures for large financial markets in discrete time. \emph{Mathematical Methods of Operations Research}, 58(3), 401-415 (2003).
	
	\bibitem{rasonyi2018utility}
	Rásonyi, M., and Meireles‐Rodrigues A. On utility maximization under model uncertainty in discrete‐time markets. \emph{Mathematical Finance} 31.1 (2021): 149-175.
	
	\bibitem{riedel2015financial}
	Riedel, F. Financial economics without probabilistic prior assumptions. \emph{Decisions in Economics and Finance}, 38(1), 75-91 (2015).

	\bibitem{schachermayer1992}
	Schachermayer, W. A Hilbert space proof of the fundamental theorem of asset pricing in finite discrete time. \emph{Insurance: Mathematics and Economics}, 11(4), 249-257 (1992).
	
	\bibitem{Schechter1996handbook}
	Schechter, E. Handbook of Analysis and its Foundations. Academic Press (1996).
	
	\bibitem{schied06}
	Schied, A. Risk measures and robust optimization problems. \emph{Stochastic Models}, 22(4), 753-831 (2006).
	
	\bibitem{stricker} 
	Stricker, C. Arbitrage et lois de martingale. In \emph{Annales de l'IHP Probabilités et statistiques} 26(3), 451-460 (1990).
	
	\bibitem{tao}
	Tao, T. An introduction to measure theory (Vol. 126). {\em American Mathematical Soc.} (2011).
	
	\bibitem{gordan}
	Žitković, G.  Convex compactness and its applications. \emph{Mathematics and Financial Economics}, 3(1), 1-12 (2010).
\end{thebibliography}


\end{document}